\documentclass[format=sigconf]{acmart}

\usepackage[compact]{titlesec}
\usepackage{enumitem}
\usepackage{amsmath, amssymb, mathtools}

\usepackage{caption}
\usepackage[llparenthesis, rrparenthesis]{stmaryrd}

\usepackage{url}

\usepackage{mathpartir}
\newcommand{\inference}[2]{\inferrule{#1}{#2}}

\usepackage{color}

\usepackage{listings}
\usepackage{hyperref}

\usepackage{xspace}

\usepackage{tikz}

\usepackage{bussproofs}

\usepackage[justification=centering]{caption} 

\usepackage[toc,page]{appendix}

\setenumerate{leftmargin=15pt}
\setitemize{leftmargin=10pt}

\newtheorem{assumption}{Assumption}[section]

\definecolor{mygreen}{rgb}{0,0.6,0}

\newif\ifprintcomments

\newcommand{\sref}[1]{\S\ref{sec:#1}}
\newcommand{\flabel}[1]{\label{fig:#1}}
\newcommand{\fref}[1]{Figure~\ref{fig:#1}}

\newcommand{\ie}{\emph{i.e.}, }
\newcommand{\eg}{\emph{e.g.}, }

\newcommand{\etc}{\emph{etc.}\xspace}
\newcommand{\Iff}{iff}
\newcommand{\etal}{\emph{et.~al.}}

\newcommand{\KW}[1]{\textsf{\textbf{#1}}}

\newcommand{\nil}{\KW{nil}}

\newcommand{\Var}[1]{\textsf{\textit{#1}}}

\newcommand{\NT}[1]{\textsf{\textit{#1}}}

\newcommand{\T}[1]{\textsf{\textbf{#1}}}

\newcommand{\hole}{[\![\;]\!]}
\newcommand{\plug}[2]{#1[\![ #2 ]\!]}
\newcommand{\mtch}[2]{\vdash_{\mathit{mtch}} #1 \approx #2}

\newcommand{\sep}{\textbar\;}

\newcommand{\dom}[1]{\ensuremath{\mathsf{dom}(#1)}}

\newcommand{\luastep}{\overset{\mathsf{\scriptscriptstyle L}}{\mapsto}}
\newcommand{\gcstep}{\overset{\mathsf{\scriptscriptstyle GC}}{\mapsto}}
\newcommand{\finstep}{\overset{\mathsf{\scriptscriptstyle F}}{\mapsto}}

\newcommand{\luagcstep}{\overset{\mathsf{\scriptscriptstyle L+GC}}{\mapsto}}

\newcommand{\rcheq}{\overset{\mathsf{rch}}{\sim}}

\newcommand{\setmetatable}{\T{setmetatable}}

\newcommand{\finlt}{\ensuremath{<^{\mathit{fin}}}}
\newcommand{\finle}{\ensuremath{\leq^{\mathit{fin}}}}

\newcommand{\gctool}{\textbf{LuaSafe}}
\newcommand{\safeset}{P_{\mathit{safe}}}
\newcommand{\typed}[1]{#1_{\mathit{typed}}}
\newcommand{\rchdef}[1]{#1_{\mathit{rch\_def}}}
\newcommand{\emptup}{\textbf{()}}
\newcommand{\typee}{\vdash_{\mathit{te}}}
\newcommand{\types}{\vdash_{\mathit{ts}}}

\setcopyright{none}
\renewcommand\footnotetextcopyrightpermission[1]{}

\begin{document}

\title{Understanding Lua's Garbage Collection}
\subtitle{Towards a Formalized Static Analyzer}

\author{Mallku Soldevila}
\affiliation{
  \department{FAMAF, UNC and CONICET}
  \country{Argentina}
}
\email{mes0107@famaf.unc.edu.ar}

\author{Beta Ziliani}
\affiliation{
  \department{FAMAF, UNC and CONICET}
  \country{Argentina}
}
\email{beta@mpi-sws.org}

\author{Daniel Fridlender}
\affiliation{
  \department{FAMAF, UNC}
  \country{Argentina}
}
\email{fridlend@famaf.unc.edu.ar}

\begin{abstract}
  We provide the semantics of garbage collection (GC) for the Lua programming
  language. Of interest are the inclusion of \emph{finalizers} (akin to
  destructors in object-oriented languages) and \emph{weak tables} (a particular
  implementation of weak references). The model expresses several aspects
  relevant to GC that are not covered in Lua's documentation but that,
  nevertheless, affect the observable behavior of programs.

  Our model is mechanized and can be tested with real programs. Our
  long-term goal is to provide a formalized static analyzer of Lua
  programs to detect potential dangers. As a first step, we provide a
  prototype tool, \gctool, that typechecks programs to ensure their
  behavior is not affected by GC. Our model of GC is validated in practice by the
  experimentation with its mechanization, and in theory by proving several 
  soundness properties.


\end{abstract}

\maketitle


\section{Introduction}
\label{sec:introduction}

Lua is an extensively used imperative scripting language.  Its
popularity grows to the point that it currently has several
interpreters and compilers~\cite{lua-implementations}, and static
analyzers~\cite{lua-analyzers}. Among its advocates, Lua has a long
standing support within the game industry~\cite{lua-games}. However,
while being a very fast scripting language, it is noted in
\emph{ibid} that:

\begin{center}\it
``Using Lua on performance-constrained platforms \\
can definitely be a challenge if you don't understand\\
 the ins and outs of Lua's memory usage.''  
\end{center}

In particular, Lua's garbage collector (GC) offers a rich interface to
let the developer efficiently deal with memory. For instance, it is
possible to create a \emph{weak table}, that is, a Lua \emph{table}
(akin to a JavaScript's' associative array) whose keys or values are
weak references. Thus, when performing garbage collection (also noted as GC),
it might decide to collect keys or values from a weak table, even if
the table is still in scope.

\begin{figure}
\begin{lstlisting}
local t = {}                     --create an empty table
setmetatable(t, {__mode = 'v'}) --set its values as weak
t[1] = {}                --assign an empty table to key 1
local i = 0
while true do
  i = i + 1
  ...           --some code, possibly generating garbage
  if not t[1] then break end
end

return i                --this value cannot be predicted
\end{lstlisting}
\vspace{-1em}
  \caption{A non-deterministic program using a weak table.}
  \label{fig:ex-non-det}
\end{figure}

If improperly used, weak tables can easily break the program's
invariants, as the simple program listed in \fref{ex-non-det}
shows. In this program, a table \textsf{t} is created containing only
one value, another table referred by a weak reference, and without any
other variable bound directly or indirectly to it. That is, there is
no other path to the value using only \emph{strong} (\ie regular)
references. Then, such value can be GC'ed at any time, making true the
condition \mbox{\lstinline{not t[1]}} at an arbitrary number of
iterations of the loop (the \lstinline{if} breaks the loop when the
value in \lstinline{t[1]} is \lstinline{nil}). Therefore, the returned
number of iterations \lstinline{i} cannot be predicted.

Weak tables are used mainly for caching values \citep{ecwt}, and a
good use of such tables will ensure the references are valid prior to
accessing them. However, in a realistic program manually validating
every use of weak tables is error-prone and, for this reason, it is proposed in~\cite{weak-js} that weak references be only used
within the scope of a library, subject to a larger scrutiny and
testing. However, testing is due to fail given the non-deterministic
nature of GC, a problem exacerbated by specificities of the
interpreter and the platform in which the program is executed.

Therefore, we aim at performing static analysis on code to detect
ill-uses of weak tables.  In this paper we present the first steps
towards that direction: a mathematical model of Lua's GC
together with a prototype tool, \gctool, whose aim is to discover
potential sources of non-determinism (at the moment, focusing only in
GC). Our model builds on top of that from~\cite{dls}, and as such, it
can be applied to the study of real Lua programs, missing only
a handful of features from the language unrelated to GC.

The model is mechanized in PLT Redex~\cite{plt} as an extension of the
mechanization presented in~\cite{dls}. It covers weak tables and
\emph{finalizers}, the latter being functions executed when an element
is about to be disposed. Without these interfaces, we show, GC is
deterministic. But as soon as finalizers or weak tables are
considered, determinism is lost.

After understanding the intricacies of Lua's GC, we develop
\gctool. This tool combines the knowledge about weak tables together
with type inference and data-flow analysis in order to detect ill-uses
of weak tables, that could lead to non-deterministic behavior. For
instance, it rightfully rejects the program from \fref{ex-non-det}.

More concretely, our contributions are:
\begin{itemize}
\item A mathematical model of Lua's GC, including finalizers and weak tables.
\item A theoretical framework under which we can express and prove standard
soundness properties of our model. 
\item A formalization of said model in PLT Redex.
\item A prototype tool, \gctool, to help uncover potential misuses of weak 
tables.
\end{itemize}

The mechanization of the model can be downloaded from~\cite{code}, and 
\gctool{} can be downloaded from~\cite{codeluasafe}.

\section{Basics of the model}
\label{sec:model}

%

%
In this section we introduce the necessary background to understand the model of
GC that we will develop in the coming sections.  As mentioned in the 
introduction, we build our model of GC on top of the semantics of Lua presented 
in~\cite{dls}, and we refer to the cited work for details.



\subsection{A subset of Lua}
\label{sec:subset_lua}

\begin{figure}
  \begin{tabbing}
  \NT{E} ::= \= [\![\;]\!] \sep 
  \NT{E} ( \NT{e} \T{,} \ldots{} ) \sep
                \NT{v} ( \NT{E$_{l}$} ) \sep \NT{var} \T{,} ... \T{=} \NT{E$_{l}$} \\
             \> \sep \T{local} \NT{var} \T{,} ... = \NT{E$_{l}$} 
                     \T{in} \NT{s} \T{end}\\
             \> \sep \T{setmetatable}(\NT{E$_{l}$}) \sep \KW{error} \NT{E} \sep \ldots{} \\

  \NT{E$_{l}$} ::= \> \NT{v} \T{,} \ldots{} \T{,} \NT{E} \T{,} 
                     \NT{e} \T{,} \ldots{}\\

  \NT{s} ::= \> \NT{e} ( \NT{e} \T{,} \ldots{} ) \sep
                   \NT{var} \T{,} ... \T{=} \Var{e} \T{,} ...\\
             \> \sep \T{local} \NT{var} \T{,} ... = \NT{e} \T{,} ... 
                     \T{in} \NT{s} \T{end}\\
             \> \sep \T{setmetatable}(\NT{e} \T{,} \ldots{}) \sep \KW{error} \NT{e} \sep \ldots{} \\

  \NT{e} ::= \> \NT{v} \sep \NT{e} ( \NT{e} \T{,} \ldots{} )
                \sep \{ [\Var{e}] = \Var{e} \T{,} ... \}\\ 
             \> \sep \T{function} ( \NT{x} \T{,} \ldots{} ) \NT{s} \T{end}
                \sep \NT{r} \sep \ldots{}\\

  \NT{v} ::= \> \NT{number} \sep \NT{string} \sep \NT{tid} \sep \NT{cid} 
                \sep \ldots{}\\

  \NT{var} ::=  \NT{x} \sep \NT{e} [\NT{e}]
  \end{tabbing}
  \vspace{-1em}
  \caption{Syntax of evaluation contexts, statements, expressions and values.}
  \flabel{syntax}
\end{figure}

\fref{syntax} shows an extract of the syntax of the model of Lua on which we are
basing our studies. A Lua program is a statement (\NT{s}), for instance a function call; multiple-variable assignment;
definition of multiple local variables; the
primitive \setmetatable{} (discussed below); \emph{error
objects}, Lua's representation of errors; and several others.
%
As expressions (\NT{e}) we have values (\NT{v}); function calls; table constructors;
\T{function} definitions; and references (\NT{r}) to the value store (to
be explained below). Values are numbers, strings, table
identifiers (\NT{tid}) and closures identifiers (\NT{cid}) (also introduced 
below).

To model imperative variables we include a mapping $\NT{r} \rightarrow \NT{v}$,
the \emph{values storage}, denoted with $\sigma$. Tables and closures are also manipulated by reference, although to ease the model we create two different sets of  \emph{identifiers} for them (\NT{tid}
and \NT{cid}). These identifiers refer to tables and closures, respectively, through a 
new mapping, the \emph{objects storage}, denoted with $\theta$. We must point out a difference from the model in~\cite{dls}: they do not consider references to closures, which in our model are required to faithfully record the cleaning of weak tables (\sref{weak_tables}).


Together with the given terms we include the corresponding \emph{evaluation
  contexts} (\NT{E}): terms with a special position marked by $[\![\;]\!]$, a
\emph{hole}. They can be used to formalize many context-dependent concepts, but
the ones shown here indicate a call-by-value execution of programs, with a 
left-to-right ordering in the arguments of sub-expressions.  We will
explain later how they are used to actually impose a particular order of
execution.

The semantics given is operational and is formalized as a relation, which we
will denote with $\luastep$, over configurations of the 
form $\sigma : \theta : \Var{s}$. For instance, the following rule formalizes
a function call:
\begin{mathpar}
 \inference {
   \mathsf{
  \theta(\Var{cid}) = \KW{function}
    \;(\Var{x}_1,...,\Var{x}_n)\; \Var{s}\; \KW{end}}\\
  \mathsf{
    \sigma' = (\Var{r}_1, \Var{v}_1),...,(\Var{r}_{\textsf{n}}, \Var{v}_n), 
              \sigma}
}
{\mathsf{
    \sigma\;\textbf{:}\; \theta\;\textbf{:}\; \Var{cid} \;
    (\Var{v}_1,...,\Var{v}_{n})
    \luastep 
    \sigma'\;\textbf{:}\; \theta\;\textbf{:}\;
    \Var{s}[\Var{x}_1 \backslash \Var{r}_1,..., \Var{x}_{n}
    \backslash \Var{r}_{n}]}}
\end{mathpar}

A function call essentially involves the allocation of its arguments
into the values' store $\sigma$, with fresh references
$\mathsf{\Var{r}_1,...,\Var{r}_n}$, and the substitution of the formal
parameters of the function by these references, in the function's body
($\mathsf{\Var{s}[\Var{x}_1 \backslash \Var{r}_1,..., \Var{x}_{n} \backslash
  \Var{r}_{n}]}$). Note that the closure is referred by its identifier. 

The following rule models the fact that the execution of a statement might happen inside a larger program, modeled with the context \Var{E}:
\vspace{-1em}
\begin{mathpar}
\inference{
  \mathsf{
    \sigma\;\textbf{:}\; \theta\;\textbf{:}\; \Var{s}
    \luastep \sigma'\;\textbf{:}\; \theta'\;\textbf{:}\; 
\Var{s'}}
  }{
  \mathsf{
    \sigma\;\textbf{:}\; \theta\;\textbf{:}\; \Var{E}[\![ \Var{s}]\!]
    \luastep
    \sigma'\;\textbf{:}\; \theta'\;\textbf{:}\; \Var{E}[\![\Var{s'}]\!]}}
\end{mathpar}
\vspace{-.5em}

\noindent The pattern from in the left of
$\luastep$ indicates that the program can be decomposed 
into an
evaluation context \Var{E} and a statement \Var{s}. If the evaluation contexts
and the execution rules are well defined, there should be just one way of
decomposing any program into an \Var{E} and an \Var{s}, and \Var{s} must be an execution-ready statement (for instance, the one presented above
for function calls). The position of the term is determined by the hole of each
evaluation context and, as can be seen in \fref{syntax}, it is unique.

\subsection{Metatables}
Lua presents a powerful metaprogramming mechanism that allows for the 
modification of the behavior of some operations under unexpected circumstances,
like arithmetic operations applied with non-numeric arguments; function
calls over non-function values; indexing a table with a nonexistent key; \etc 
At the heart of this mechanism lies the concept of \emph{metatable}, a regular 
table that maintains handlers to manage unexpected situations, associated with 
specific keys defined beforehand. For instance, in order to \emph{explain} how 
a given table should be represented as a string, through the service 
\textsf{tostring}, the developer can associate a conversion function with the 
key ``\textsf{\_\_tostring}'' in the table's metatable. Some type of objects 
(tables and \emph{userdata}) allows for the definition of a single metatable per 
value, while for the remaining there is just one metatable per type.

A table can be set a metatable through the \setmetatable{} library
service. In~\cite{dls}, tables are modeled as a pair containing
the table's data and a table identity for the metatable, which can be \nil.
As an example, the following rule specifies the creation of a table:

\vspace{-1em}
\begin{mathpar}	
  \inference{\mathsf{\Var{tid} \notin dom(\theta_1)}\\
    \mathsf{\theta_2 = (\Var{tid}, \;
                       (\{[\;\Var{v}_1\;]\; =\;\Var{v}_2,\;... \}\;, 
                       \nil)),\;
                       \theta_1}}
  {\mathsf{\sigma \; \KW{:} \;
      \theta_1\; \KW{:} \;
      \{ [\;\Var{v}_1\;]\; =\;\Var{v}_2, \;... \}
      \;\luastep\;
      \sigma \; \KW{:} \;
      \theta_2\; \KW{:} \;
      \Var{tid}}}
\end{mathpar}
\vspace{-.5em}

\noindent After creation, a table does not contain a metatable set. Only trough
\setmetatable{} one can associate a metatable to the given table.

As we will see in the coming section, metatables play an important
role in the semantics of GC
.

\section{Garbage collection}
\label{sec:applications}


This section represents our main contribution: an abstract model of Lua's GC, modularly divided in three parts. We start by modeling GC without interfaces (\sref{gc:simple}), laying 
the basic concepts upon which the interfaces are added: finalizers (\sref{finalizers}), and  weak tables (\sref{weak_tables}).


\subsection{Reachability-based garbage collection.}
\label{sec:gc:simple}

Lua implements two reachability-based GC strategies: a \emph{mark-and-sweep}    
collector (the default) and a \emph{generational collector}. 
The user is entitled to change the algorithm by calling the 
\textsf{collectgarbage} standard library's service.
In this section we will provide a specification for the behavior of a typical 
reachability-based GC. It should encompass the essential details
of the behavior of the two algorithms included in Lua and any other based on 
reachability. We start with a small set of definitions that we will enrich in coming sections.

\paragraph{Reachability.}

The purpose of GC is to remove from memory (the store) information that will not be used
by the remaining computations of the program. One of the simplest and commonly
used approaches to find such information is based on the notion of
\emph{reachability}~\citep[\eg\!\!][]{acjro}. The idea is simple: given the set of references that
literally occur in the program (the \emph{root set}), it must be the case that
any \emph{information} (\eg value in a store) that may be used by the program
must be \emph{reachable} from that set. Conversely, any \emph{binding} (a
reference with its value) in the store that cannot be reached from the root set,
will not be accessible from the program and, therefore, can be safely removed as
it will not be needed in the remaining computations of the program.

In the context of this work, those values which are not reachable will be called
\emph{garbage}. This notion, sufficient to model Lua's GC, is purely syntactic: it will
take into account just the literal occurrence of references in the program, or
their reachability from this set of references that occur literally,
to determine if a given value is garbage or not. In contrast, there are 
approaches, to identify garbage, where also the semantics of the program may
be taken into account~\citep[\eg\!\!][]{fswr}.

To formally capture the notion of garbage, it will be easier to begin with the
definition of reachable references. The only difference worth to mention, in 
comparison with common definitions found in the literature~\cite{fsf,acjro},
is the inclusion of metatables: a metatable of a reachable table is considered 
reachable, so a reachability \emph{path}, that is, a path between a reference and the root set, might also go through a metatable.

Informally, a location (value reference or an identifier) will be reachable
with respect to a given term \Var{t}, and corresponding stores, if one of the
following conditions hold:
\begin{itemize}
  \item The location occurs literally in \Var{t}.
  \item The location is reachable from the information associated with a
    reachable location. This includes:
    \begin{itemize}
    \item The location is reachable from the closure associated with a
      reachable location.
    \item The location is reachable from the table associated with a
      reachable location.
    \item The location is reachable from a metatable of a reachable table
identifier.
    \end{itemize}
\end{itemize}

This is formalized in the following definition:

\begin{definition}[Reachability for Simple GC]
\label{Reachability Simple GC}
We say that a location $\Var{l} \in \Var{r} \cup \Var{tid} \cup \Var{cid}$ is 
\emph{reachable} in term \Var{t}, given stores $\sigma$ and $\theta$, \Iff:
\begin{tabbing}
  $\mathsf{reach(\Var{l}, \Var{t}, \sigma, \theta)}$ = \=
  $\mathsf{\Var{l} \in \Var{t} \; \vee}$\\
  \> $\mathsf{(\exists \Var{r} \in \Var{t},}$ 
  $\mathsf{reach(\Var{l}, \sigma(\Var{r}),
    \sigma \setminus \Var{r} , \theta)) \; \vee}$\\
  \> $\mathsf{\exists \; \Var{tid} \in \Var{t},}$ \= ($\mathsf
  {reach(\Var{l}, \pi_1(\theta(\Var{tid})), \sigma,
    \theta \setminus \Var{tid}) \; \vee}$\\
  \>\>~$\mathsf{reach(\Var{l}, \pi_2(\theta(\Var{tid})), \sigma,
    \theta \setminus \Var{tid})) \; \vee}$\\
  \> $\mathsf{\exists \; \Var{cid} \in \Var{t},}$ \= $\mathsf
  {reach(\Var{l}, \theta(\Var{cid}), \sigma,
    \theta \setminus \Var{cid})}$
\end{tabbing}
\end{definition}

We write $l \in t$ to indicate that $l$ occurs literally in term $t$, and write
$\gamma \setminus l$ as the store obtained by removing the binding of $l$ in
$\gamma$. Informally, this predicate states that either \Var{l} occurs in 
\Var{t}, or there is a reference in \Var{t} such that \Var{l} is reachable 
from it.

To avoid cycles generated from mutually recursive definitions, in the stores,
 that would render undefined the preceding 
predicate, we remove from the stores the bindings already considered. We assume 
the predicate is false if a given location occurs in \Var{t} but does not 
belong to the domain of any of the stores.

Note that for a table \NT{tid} we not only check its content 
($\pi_1(\theta(\Var{tid}))$) but also its metatable ($\pi_2(\theta(\Var{tid}))$).
That is, a table's metatable is considered reachable when the table itself is 
reachable.
Observe that, being metatables ordinary tables, they can contain other tables' 
ids or even closures, which in turn may have other locations embedded 
into them. Naturally, if metatables were not taken into account for reachability,
we could run straight into the problem of dangling references any time a 
metamethod is recovered from the metatable. Also, note that during the recursive 
call
$\mathsf{reach(\Var{l}, \pi_2(\theta(\Var{tid})), \sigma, \theta \setminus 
\Var{tid})}$, at first it will determine if \Var{l} is exactly
$\mathsf{\pi_2(\theta(\Var{tid}))}$ (because it asks for 
$\Var{l} \in \pi_2(\theta(\Var{tid}))$, for $\pi_2(\theta(\Var{tid}))$ being 
either \nil{} or a table identifier) and, if not, it will continue with the 
inspection of the content of the metatable, by dereferencing its id, given that
it is not \nil{}. Hence, we do not remove
$\pi_2(\theta(\Var{tid}))$ from $\theta$ in the mentioned recursive call.

The last disjunct checks for reachability following a closure 
identifier \Var{cid} present in the root set of references. We need to expand the reachability tree following the environment of the closure 
(\ie the mapping between the external variable's identifiers, present
in the body of the closure, and their corresponding references). 

We conclude this part on reachability with a minor observation:
naturally, the reference manual leaves unspecified details of GC. For
instance, it does not mention how metatables affect GC even though it
does have an observable effect on programs. One of our major
challenges and aim in this work is to unveil such interactions.

\paragraph{Specification of a garbage collection cycle.}

We keep abstract the specification of a cycle of GC in order to accommodate to any implementation of GC:

\begin{definition}[Simple GC cycle]~
\label{Simple GC cycle}\\
\indent $\mathsf{gc(\Var{s}, \sigma, \theta) = (\sigma_1, \theta_1)}$, where:
\begin{itemize}
  \item[-] $\mathsf{\sigma = \sigma_1 \uplus \sigma_2}$
  \item[-] $\mathsf{\theta = \theta_1 \uplus \theta_2}$
  \item[-] \begin{tabbing}
           $\mathsf{\forall l}$ \= 
           $\mathsf{\in dom(\sigma_2) \cup dom(\theta_2),
           \neg reach(l, \Var{s}, \sigma, \theta)}$
           \end{tabbing}  
\end{itemize}
\end{definition}
We use $\gamma_1 \uplus \gamma_2$ to denote the union of stores with disjoint 
domains. This specification states that $\mathsf{gc}(\Var{s}, \sigma, 
\theta)$ returns two stores, $\mathsf{\sigma_1}$ and $\mathsf{\theta_1}$, which 
are a subsets of the stores provided as arguments, $\mathsf{\sigma}$ and 
$\mathsf{\theta}$. We do not specify how these subsets are determined. We just 
require that the remaining part of the stores ($\mathsf{\sigma_2}$ and 
$\mathsf{\theta_2}$) do not contain references that are reachable from the 
program \Var{s}.
Satisfied this condition, it is safe to run code \Var{s} in the new stores 
$\mathsf{\sigma_1}$ and $\mathsf{\theta_1}$, as no dereferencing of a dangling 
pointer may occur.

Observe that the previous specification does not impose 
$\sigma_1$ and $\theta_1$ to be \emph{maximal}, meaning they might have non-reachable references with respect to \Var{s}.


Using the previous specification of GC, we can extend our model of Lua with a 
non-deterministic step of GC, through a relation $\gcstep$:
$$
  \inference
  {\mathsf{(\sigma', \theta') = gc(\Var{s}, \sigma, \theta)}\\
    \mathsf{\sigma' \neq \sigma \vee \theta' \neq \theta}\\
  }
  {\mathsf{\sigma : \theta : s \gcstep \sigma' : \theta' : s}}
$$
We require it to actually perform some change to the stores to ensure progress. This obviously introduces non-determinism:
at any time, as long as there is some garbage left, we can choose to collect the
garbage or to continue with the execution of the program. But, for the
definition provided so far, this non-determinism should not change the behavior
of the program: every execution path will eventually lead to the same result. We
will define formally this concepts in \sref{properties_gc}. This property will 
not longer be true when extending GC with finalizers and weak tables.


\subsection{Finalizers.}
\label{sec:finalizers}

Lua implements finalizers, a mechanism commonly present in programming languages with GC, useful
for helping in the proper disposal of external resources used by the
program. They are defined by the programmer as a function, which is called by
the garbage collector after a value amenable for finalization (table or
userdata) becomes garbage. It should be noted that because finalizers are called
by the garbage collector, there is no possibility of determining the precise 
moment in which finalization will occur. This in contrast with 
\emph{destructors}, a concept present in languages with deterministic 
memory management algorithms (\eg as in C++).

There are several problems that arise from the misuse of this mechanism, 
associated with the fact that finalizers are called in a non-deterministic fashion, 
introducing that non-determinism into the execution of the program. Nonetheless, 
the implementation of finalizers in Lua provides some guarantees about
the execution order of finalizers and the treatment given to resurrected objects 
which makes the algorithm an interesting case study.


\subsubsection{Overview of finalizers in Lua.}

We will begin with an informal presentation of the semantics of finalizers in Lua. After this, we will show how to extend the previous model of GC to include
this interface with the garbage collector.

\paragraph{Setting up a finalizer.} 
\begin{figure}
\begin{lstlisting}
local a, b = {}, {}
setmetatable(a, b) !\label{code:setmeta}!
b.__gc = function () print("bye") end  !\label{code:__gc}!
a = nil
collectgarbage()            --nothing is printed !\label{code:collect}!
local c = {}
setmetatable(c, b)                               !\label{code:newset}!
b.__gc = function () print("goodbye") end
c = nil
collectgarbage()      --now it outputs 'goodbye'  !\label{code:goodbye}!
b.__gc = "not a function"   !\label{code:notafun}!
local d = {}
setmetatable(d, b)
d = nil
collectgarbage()               --nothing happens  !\label{code:nofail}!
\end{lstlisting}
\vspace{-1em}
\caption{Setting up a finalizer.}\label{fig:finalizer}
\end{figure}

The finalizer of an object (table or userdata) is a function stored in
the object's metatable, associated with the key
\textsf{``\_\_gc''}. For \emph{finalization} to occur (\ie the
execution of the finalizer) the key must be defined the first time the
corresponding metatable is set.  In that case, it is said that the
given object is \textit{marked} for finalization. Later definitions of
the \textsf{\_\_gc} field will not be considered. The code shown in
\fref{finalizer} shows this behavior: when \textsf{a} is set an empty
metatable (\textsf{b} in Line~\ref{code:setmeta}), even if later on
\textsf{\_\_gc} is defined (Line~\ref{code:__gc}), when \textsf{a} is
garbage collected (Line~\ref{code:collect}), no output is
produced. But now that \textsf{b} has the \textsf{\_\_gc} field
defined, when it is set as a metatable of a new object
(Line~\ref{code:newset}), this object is correctly marked for
finalization (Line~\ref{code:goodbye}). Also, if the value set in the
field \textsf{\_\_gc} is not a function, GC will simply silently
ignore the error (lines~\ref{code:notafun} to~\ref{code:nofail}). As
a last remark, the last finalizer set, assuming it is a function, is
the one called when the object is disposed.

\paragraph{Execution order of finalizers}
\begin{figure}
\begin{lstlisting}
local a, b = {}, {}
local c = {__gc = function (o) print("bye", o) end}
print(a, b)          --table: 0x56..00 table: 0x56..40
setmetatable(a, c)!\label{code:seta}!
setmetatable(b, c)!\label{code:setb}!
a, b = nil, nil
collectgarbage()
      --bye table: 0x56..40 (b) bye table: 0x56..00 (a) !\label{code:byebye}!
\end{lstlisting}
\vspace{-1em}
\caption{Chronological order of execution of finalizers.}
\label{fig:finalizer:order}
\end{figure}

The execution order of finalizers is chronologically inverse to the
time of the definition of the finalizers. This behavior is explained
in \fref{finalizer:order}. This code performs the following steps: 1)
creates two tables, \textsf{a} and \textsf{b}; 2) sets a
metatable \textsf{c} to these objects containing a finalizer that prints the object being
finalized, first for \textsf{a} and then for \textsf{b}; 3) eliminates
any reference to \textsf{a} and \textsf{b}; and 4) invokes the garbage
collector. As you can see from the output (Line~\ref{code:byebye}),
the order in which the metatable is set affects the order in which the
finalizers are called. While not shown in the code, if we swap
lines~\ref{code:seta} and~\ref{code:setb}, the result will also be
swapped.


\paragraph{Resurrection.}

During finalization of a given object, its location is passed to the finalizer,
turning the object reachable again. This phenomenon is commonly known as 
\textit{resurrection}, and is normally transient.
Then, there exist the possibility that the user code of the finalizer makes 
permanent the resurrection, by creating an external reference to the object, 
turning it reachable again even after finalization, preventing it from being 
collected. 

This possibility introduces problems \cite{fincolint} into the 
implementation of garbage collectors, reduces their 
effectiveness to reclaim memory unused by the program and could reintroduce 
into the program objects that do not satisfy representation invariants.

To mitigate this issue, Lua treats finalized objects specially: it does not allow for a 
finalized object to be marked again for finalization.
In that way, the finalizer of an object will never be called twice, avoiding 
indestructible objects. The object will be destroyed once it becomes unreachable 
again.
This is the only difference of a finalized object: it is still possible to set a new metatable and to configure the 
resurrected objects' behavior using every metamethod \emph{but} ``\textsf{\_\_gc}''.

\paragraph{Error handling.}

During execution of a program, any error in a finalizer is 
propagated to the main thread of execution. Because finalizers are interleaved 
with user code, any error thrown from a finalizer appears in a position in 
the program that cannot be determined in advance. If that position happens to 
be inside a function that was called in \emph{protected mode} ---like a
\textsf{try} in other languages--- then the error is caught.

When a program ends normally, Lua executes each finalizer of the
remaining objects in protected mode.  In that circumstance, any error
occurred during the execution of a given finalizer, interrupts only
that finalizer, allowing for the call of the remaining
finalizers. Also, a finalizer ended by an erroneous situation does not
prevent the corresponding object from being disposed.


\subsubsection{Modeling finalizers.}
\label{sec:mod_finalizers}

We extend the model to include finalizers in two steps: first we update the
internal representation for tables presented in \sref{model} to add information about
finalizers; then, we modify the GC model introduced in \ref{sec:gc:simple}, to 
be aware of the finalization mechanism. 

\paragraph{Representation of tables.}

We extend the tuple for representing a table with a third field, obtaining 
\textsf{(table, metatable, pos)}. The new field \textsf{pos} has three 
different possible values: if it is $\bot$, it means that there is no finalizer 
set for the table; if it is $\oslash$, it means that the table cannot be set 
for finalization (to avoid multiple resurrections);
and if it 
is a value $p$, of a set of values $\mathcal{P}$ ordered by a given order
\finlt, it means the finalizer is set, with priority $p$, according to  
\finlt. Initially, 
\textsf{pos} will be $\bot$, as shown in the first rule of \fref{expr_theta_stores_obj_creat}.
We present its semantics (and the remaining computation rules for finalization), 
with a new relation, $\finstep$.

\begin{figure}
\begin{mathpar}	
  \inference{\mathsf{\forall\;1 \leq i, \Var{field}_i = \Var{v} \vee
      \Var{field}_i = [\;\Var{v}\;]\; =\;\Var{v}'}\\
    \mathsf{\theta_2 = (\;\Var{tid}, \;
                       (\;addkeys(\{\Var{field}_1 ,\;... \})\;, 
                       \;\nil\colorbox{lightgray}{$\mathsf{\!\!,\;\bot}$})),\;
                       \theta_1}}
  {\mathsf{\theta_1\;:\;\{ \Var{field}_1, \;... \}
      \;\finstep\;\theta_2\;:\;\Var{tid}}}
\and
  \inference{\mathsf{\delta(type, \Var{v}) \in \{``table", ``nil" \}}\\
    \mathsf{indexmetatable(\Var{tid} , ``\_\_metatable", \theta_1) = \nil}\\
    \mathsf{\theta_2 =
      \theta_1[\Var{tid} := (\pi_1(\theta_1(\Var{tid})), \Var{v}
      \colorbox{lightgray}{$\mathsf{\!, set\_fin(\Var{tid}, \Var{v}, \theta_1)\!}$})]}}
  {\mathsf{\theta_1 \; \textbf{:}
      \;\setmetatable(\Var{tid}, \Var{v})
      \finstep
      \;\theta_2\; \textbf{:}\; \Var{tid}}
  }
\end{mathpar}
\vspace{-1em}
\caption{Selected rules extended with finalization.}
\flabel{expr_theta_stores_obj_creat}
\end{figure}

As mentioned, \finlt{} is defined chronologically by 
the moment in which an object has been marked for finalization. For our semantics, it suffices to have a 
function \textsf{next} with signature $\mathcal{P} \rightarrow \mathcal{P}$, 
which should provide an element of $\mathcal{P}$ larger than its argument. We 
will also need $\bot \in \mathcal{P}$, and to be minimum with respect to 
\finlt. When a metatable is set with the corresponding call to 
\setmetatable{} (second rule of \fref{expr_theta_stores_obj_creat}), we  use a helper function \textsf{set\_fin} to compute  the corresponding value of \textsf{pos}. 

\newcommand{\interspc}{0em}

\begin{figure}
\begin{align}
\mathsf{set\_fin(\Var{tid}, \Var{v}, \theta)} &
= \oslash, \hspace{30pt} \textit{if} \; \mathsf{\pi_3(\theta(\Var{tid})) = \oslash}\\[\interspc]
\mathsf{set\_fin(\Var{tid}, \nil, \theta)} &= 
\bot, \hspace{30pt} \textit{if} \; \mathsf{\pi_3(\theta(\Var{tid})) \neq \oslash}\\[\interspc]
\mathsf{set\_fin(\Var{tid}, \Var{tid}_{m}, \theta)} &
= \pi_3(\theta(\Var{tid})), \hfill \textit{if} \; \left\{ 
  \begin{array}{ll}
  \mathsf{\pi_2(\theta(\Var{tid})) = \Var{tid}_{m}}\\
  \mathsf{\pi_3(\theta(\Var{tid})) \neq \oslash}
  \end{array}  \right.\\[\interspc]
\mathsf{set\_fin(\Var{tid}, \Var{tid}_{m}, \theta)} &
= \bot, \hspace{29pt} \textit{if} \; \left\{ 
  \begin{array}{ll}
  \text{``\_\_gc''} \notin \mathsf{\pi_1(\theta(\Var{tid}_{m}))}\\
  \mathsf{\pi_2(\theta(\Var{tid})) \neq \Var{tid}_{m}}\\
  \mathsf{\pi_3(\theta(\Var{tid})) \neq \oslash}
  \end{array}  \right.\\[\interspc]
\mathsf{set\_fin(\Var{tid}, \Var{tid}_{m}, \theta)} &
= \mathsf{next}(p), \hspace{9pt} \textit{if} \; \left\{ 
  \begin{array}{ll}
  \text{``\_\_gc''} \in \mathsf{\pi_1(\theta(\Var{tid}_{m}))}\\
  \mathsf{\pi_2(\theta(\Var{tid})) \neq \Var{tid}_{m}}\\
  \mathsf{\pi_3(\theta(\Var{tid})) \neq \oslash}
  \end{array}  \right.
\end{align}
$$\text{where} \;
p = \mathsf{max^{\finlt}(filter(map(\pi_3, img(\theta)), 
                       \lambda \; pos. pos \neq \oslash))}
$$
\caption{Function \textsf{set\_fin} for computing the \textsf{pos} field.}\label{fig:setfin}
\end{figure}

\fref{setfin} shows the \textsf{set\_fin} function, which takes two
tables (a table identifier \Var{tid} and the proposed metatable) and a
store $\theta$. The metatable is another table identifier
$\Var{tid}_\mathsf{m}$ or \nil, and returns the new \textsf{pos} value. 
The first
equation shows the main use of the $\oslash$: no matter what is the
value of the metatable, if the previous \textsf{pos} field of the
table contains an $\oslash$, then it returns $\oslash$ to ensure no
finalization can happen again on \Var{tid}. The second equation
specifies one of the situations when a given table is unmarked for
finalization: if the metatable is \nil, and the previous value of
\textsf{pos} is not $\oslash$, then it returns $\bot$. The third
equation considers the case when the same metatable is set, in which
case the \textsf{pos} field remain unchanged (we use the bracket to
mean that every condition must apply). The fourth equation considers
the case when the metatable does not contain the \textsf{``\_\_gc''}
metamethod: it is unmarked for finalization ($\bot$). In the last
equation \textsf{set\_fin} returns the \textsf{next} value of the
maximum of every \textsf{pos} in $\theta$, if the metatable contains
the metamethod \textsf{``\_\_gc''}.

\paragraph{Specification of GC with finalization.}
\begin{figure}
\begin{tabbing}
$\mathsf{gc_{fin}(\Var{s},\sigma,\theta) = (\sigma_1,\theta_{1}',\Var{t})}, \mathit{where}$\\
$\mathit{gc}
 \left\{ 
   \begin{array}{lcl}
     \mathsf{\sigma = \sigma_1 \uplus \sigma_2} \\
     \mathsf{\theta = \theta_1 \uplus \theta_2} \\
     \mathsf{\forall l \in dom(\sigma_2) \cup dom(\theta_2),
             \neg reach(l, \Var{s}, \sigma, \theta)}
   \end{array}
 \right.$\\ 

$\mathit{fin}
 \left\{ 
   \begin{array}{lcl}
     \mathsf{\forall \Var{tid} \in dom(\theta_2),}\\
     \quad \mathsf{\neg marked(\Var{tid}, \theta_2)}\\
     \mathsf{\forall \Var{l} \in dom(\sigma_2) \cup dom(\theta_2),}\\
     \quad \mathsf{not\_reach\_fin(\Var{l}, \sigma_1, \theta_1)}\\
     \small[ \mathsf{\exists \Var{tid} \in \dom{\theta_1}},\\
     \quad \mathsf{fin(\Var{tid}, \Var{s}, \sigma, \theta)}\\ 
     \quad\mathsf{next\_fin(\Var{tid}, \Var{s}, \sigma, \theta)}\\
     \quad\mathsf{\Var{v} = indexmetatable(\Var{tid}, ``\_\_gc", \theta_1)}\\
     \quad \mathsf{\Var{v} \in \NT{cid} \Rightarrow 
                   \Var{t} = \Var{v}(\Var{tid})}\\
     \quad \mathsf{\Var{v} \notin \NT{cid} \Rightarrow \Var{t} = \nil}\\
     \quad\mathsf{\theta'_{1} = \theta_1[\Var{tid} := 
                               (\pi_1(\theta_1(\Var{tid})), 
                                \pi_2(\theta_1(\Var{tid})),
                                \oslash)]}\\
   \text{or:}\\
     \quad \Var{t} = \nil\\
     \quad \theta_1' = \theta_1\small]
   \end{array}
 \right.$
\end{tabbing}
\vspace{-1em}
  \caption{GC cycle with finalization.}
  \label{fig:gcfin}
\end{figure}

\begin{figure}
  \begin{tabbing}
    $\mathsf{marked(\Var{tid}, \theta) \doteq 
   \pi_3 (\theta(\Var{tid})) \notin \{ \bot, \oslash \}}$\\[.3em]
    $\mathsf{not\_reach\_fin(\Var{l}, \sigma, \theta) \; \doteq}$ \= 
           $\mathsf{\not \exists \; \Var{tid} \in dom(\theta), \Var{l} \neq \Var{tid} \; \wedge}$\\
           \> $\mathsf{marked(\Var{tid}, \theta) \; \wedge \;
           reach(\Var{l}, \Var{tid}, \sigma, \theta)}$\\
    $\mathsf{fin(\Var{tid}, \Var{s}, \sigma, \theta)}$ \=
           $\mathsf{\doteq
             \neg reach(\Var{tid}, \Var{s}, \sigma, \theta) \wedge
             marked(\Var{tid}, \theta)}$\\[.3em]
    $\mathsf{next\_fin(\Var{tid}, \Var{s}, \sigma, \theta)}$ \= 
           $\mathsf{\doteq}$ \= 
           $\mathsf{\forall \Var{tid'} \in dom(\theta)}$,\\
           \>\> $\mathsf{fin(\Var{tid'}, \Var{s}, \sigma, \theta)
     \Rightarrow
             \pi_3(\theta(\Var{tid'})) \finle \pi_3(\theta(\Var{tid}))}$
  \end{tabbing}
\vspace{-1em}
  \caption{Predicates for finalization.}
  \label{fig:finalization_predicates}
\end{figure}

We enrich the specification of GC from \sref{gc:simple} to make it
aware of finalization (figures~\ref{fig:gcfin}
and~\ref{fig:finalization_predicates}). The new predicate,
$\mathsf{gc_{fin}}$, returns two stores $\sigma_1$ and $\theta'_1$,
and a term \Var{t}, the finalizer to be executed if appropriate. The
first part of the predicate (\textit{gc}) replicates the {\sf gc} predicate from \sref{gc:simple}, and states that we can split
the stores into two disjoint parts, the ones to be discarded
($\sigma_2$ and $\theta_2$) and the rest ($\sigma_1$ and
$\theta_1$). But now the partitions have additional conditions
(\textit{fin}): first, every discarded table \Var{tid} in $\theta_2$ must not be
marked for finalization, otherwise we will lose a call to a
finalizer. Second, we ask that every location from the removed stores
is not reachable from the stores that are kept ($\sigma_1$ and
$\theta_1$). 

The previous conditions ensure that $\theta_2$ only contain tables already finalized or not
set for finalization, and avoids potential dangling pointer
errors when executing a finalizer. The following conditions 
characterize the next table to be finalized. If there exists a
\Var{tid} in $\theta_1$ such that it is finalizable and the next in 
the order \finle{} (as expressed by the predicates {\sf fin} and 
{\sf next\_fin}), and has a proper finalizer set (a function \Var{v} in its 
``\textsf{\_\_gc}'' field), then the next statement to be executed is \Var{v} 
applied to the table identifier (transiently resurrecting the table), and the 
new table store $\theta_1'$ is the same as $\theta_1$, except that \Var{tid} 
 is forbidden to be marked again for finalization (by setting its 
\textsf{pos} field to $\oslash$), therefore avoiding more than one 
resurrection of the table. Note that \Var{tid} is still in the returned 
$\theta_1'$, otherwise it could not be made accessible to the finalizer. In 
our model, the table is actually collected in another GC cycle, as we cannot 
know before hand if it will be resurrected or not by its finalizer.

In case there is no table with a valid finalizer, then \Var{t} is \nil{} and
$\theta_1'$ is just $\theta_1$.



\paragraph{Interleaving finalization with the user program.}

\begin{figure}
  \begin{mathpar}
  \inference
  {\mathsf{(\sigma', \theta', \Var{v} (\Var{tid})) =
      gc(\sigma, \theta, E [\![ s ]\!])}
    }
  {\mathsf{\sigma : \theta : E [\![ s ]\!]\;
      \finstep
      \sigma' : \theta' : E [\![\; \Var{v} (\Var{tid}) ; s ]\!]\;
    }}
  
  \inference
  {\mathsf{(\sigma', \theta', \Var{v} (\Var{tid})) =
      gc(\sigma, \theta, E [\![ e ]\!])}
  }
  {\mathsf{\sigma : \theta : E [\![ e ]\!]\;
      \finstep
      \sigma' : \theta' : E [\![\;
      \KW{function} \; \$ \; () \;
      \KW{return} \; e \; \KW{end} \; (\Var{v} (\Var{tid}))]\!]\;
    }}

  \inference
  {\mathsf{(\sigma', \theta', \nil) =
      gc(\sigma, \theta, \Var{s})}\\
      \mathsf{\sigma' \neq \sigma \vee \theta' \neq \theta}
  }
  {\mathsf{\sigma : \theta : \Var{s}
      \finstep
      \sigma' : \theta' : \Var{s}
    }}
\end{mathpar}
\vspace{-1em}
  \caption{Interleaving the execution \\of finalizers with the program.}
  \label{fig:interleaving}
\end{figure}

From the definition of $\mathsf{gc_{fin}}$ given above, it is clear that a single GC 
cycle encompasses collection of garbage together with at most one call to a  
finalizer. The reasons are two-fold: first, the small-step fashion of our 
semantics, and the interleaved execution of finalizers with the user's program. However, 
this does not prevent the execution of more than one finalizer before the 
execution of the next user program's instruction, given the non-determinism of the
execution rules for GC.

What remains to specify is how finalization is actually interleaved with the
user program. This is stated by the rules in \fref{interleaving}.
We allow for the possibility of interleaving the finalization step with any
statement or expression to be executed. The first case can be expressed
directly, as shown in the first rule. Interleaving it with an expression, shown
in the second rule, requires some more work, since we cannot express directly
the concatenation of expressions. In that case, we reduce the desired execution
order of expressions to the one defined for function call. 

Finally, if no finalizer is chosen (third rule), as before, we ask for some of the stores
returned to be modified in order disallow infinite sequences of GC steps.

\subsection{Weak tables}
\label{sec:weak_tables}


A \emph{weak table} is a table whose keys and/or values are referred by 
\emph{weak references}: references which are not taken into account by the 
garbage collector when determining reachability. In Lua, among the types 
included into our model, only tables and closures can be garbage collected from 
weak tables, the general rule being 
that \textit{``only objects that have an explicit construction are removed from 
weak tables''} \citep[\S2.5.2 of][]{lua-refman}.

In order to specify a table's weakness, the user adds in the table's metatable 
the key \textsf{``\_\_mode''} with a string value containing the characters 
\textsf{`k'} (for keys to be referred by weak references) and/or \textsf{`v'}
(for values to be referred by weak references) .

\paragraph{Introducing weak tables into the model.}

To model weak tables we do not introduce weak references
explicitly. Instead, we modify the criterion used to determine the
reachability of a given reference to consider its occurrences on weak
tables, according to the tables' weakness. Key to the new definition of GC cycle is a new predicate {\sf reachCte} that allows us to consider the reachability of a \emph{collectible table element} (\NT{cte}), which is an element of the set with the same name formed from the union of table
and closures identifiers.

\paragraph{Reachability of a \NT{cte}.}

We distinguish two situations with respect to the reachability of a 
\NT{cte}: either there is a path from the root set of references to the value 
itself using just \emph{strong} references (non-weak references), or every 
path to the value from the root set contains at least one weak reference. In the first case the value will not be garbage collected, and we 
refer to such value as \emph{strongly reachable}. In the second case the value can be GC.

In order to distinguish these cases, we define what are a table's \emph{strong occurrences} (\fref{so}): the keys
and/or values of a table (limited to \NT{cte}s) that are not referred
by weak references.
\begin{figure}
  \begin{tabbing}
  $\mathsf{SO(\Var{tid}, \theta)}$ \= =
  $ \left\{ 
    \arraycolsep=0pt\def\arraystretch{0.8}
    \begin{array}{lcl}
              \setlength{\tabcolsep}{12cm}
              \mathsf{\{ \Var{k}_i | \Var{k}_i \in (\{ \Var{k}_1,... \}
              \cap \NT{cte}) \}}
              & \textit{if} & \mathsf{wv?(\Var{tid}, \theta)}\\
              &&     \;\;\wedge\; \mathsf{\neg wk?(\Var{tid}, \theta)}\\[0.3em]
              \mathsf{\{ \Var{v} | \Var{v} \in \{ \Var{k}_1, \Var{v}_1 ,... \} 
              \cap \NT{cte} \}}
              
              & \textit{if} & \mathsf{\neg (wv?(\Var{tid}, \theta)}\\
              &&     \;\;\;\vee\; \mathsf{wk?(\Var{tid}, \theta))}\\[.3em]
              \mathsf{\{ (\Var{k}_i, \Var{v}_i) | \Var{v}_i 
              \in \{ \Var{v}_1,... \}
              \cap \NT{cte} \}} 
              & \textit{if} & \mathsf{\neg wv?(\Var{tid}, \theta)}\\
              &&     \;\;\wedge\; \mathsf{wk?(\Var{tid}, \theta)}\\[.3em]
              
              \mathsf{\emptyset}
              & \qquad & \hspace{-12pt}\textit{otherwise}
            \end{array}
          \right.$\\[.3em]
          $\text{where} \;
          \mathsf{\pi_1(\theta(\Var{tid})) = 
            \{ [\Var{k}_1] = \Var{v}_1,... \}}$
\end{tabbing}
\vspace{-1em}
  \caption{Strong occurrences of a table.}
  \label{fig:so}
\end{figure}
If a given table has weak values then just its keys' occurrences are
considered strong (predicates \textsf{wk?} and \textsf{wv?}, elided
for brevity, allow us to know the weakness of a given table). The
second and fourth cases can be explained on the same basis. The third
case, weak keys and strong values, has to do with what is known as an
\emph{ephemeron} table, which is treated in a special way by the
garbage collector, in order to avoid the problems that arise with
cycles into a weak table (\eg values referring to their own keys),
which could prevent them from proper GC, or between weak tables with
this level of weakness, which could delay GC (see \cite{ecwt} for an
analysis of the problem from Lua's perspective). In an ephemeron table, an occurrence of a
value from \NT{cte} as the value of a table field is considered strong
just if its associated key is still strongly reachable. Because this
is not a property that can be determined locally, by just looking at
the table being inspected, we return each key-value pair.


\begin{figure}
  \begin{tabbing}
    $\mathsf{eph(\Var{id}, (\Var{k}, \Var{v}), \Var{tid}, \Var{rt}, \sigma, 
             \theta)}$ = \=
           $\mathsf{reachCte(\Var{id}, \Var{v}, \sigma, \theta, \Var{rt}) \; \wedge}$\\
    \>$\mathsf{[\Var{k} \notin \Var{cte} \vee
             reachCte(\Var{k}, \Var{rt}, \sigma, 
             \left.\theta\right|_{\left.\Var{tid}\right|_{\Var{k}}}, \Var{rt})]}$ 
  \end{tabbing}

\begin{tabbing}
    $\mathsf{reachTable}$\=$\mathsf{(\Var{id}, \Var{tid}, \sigma, \theta, \Var{rt})}$ =\\ 
    \>$\mathsf{[\exists}$
    $\mathsf{(\Var{k}, \Var{v}) \in SO(\Var{tid}, \theta),}$
     $\mathsf{
      eph(\Var{id}, (\Var{k}, \Var{v}), \Var{tid}, \Var{rt}, \sigma, \theta)]    
      \; \vee}$\\
    \> $\mathsf{[\exists}$
       $\mathsf{\Var{v} \in SO(\Var{tid}, \theta),}$
       $\mathsf{reachCte(\Var{id}, \Var{v}, \sigma, \theta, \Var{rt})] \; 
                \vee}$\\
    \>$\mathsf{reachCte(\Var{id}, \pi_2(\theta(\Var{tid})), \sigma, 
                          \theta, \Var{rt})}$
\end{tabbing}

\begin{tabbing}
    $\mathsf{reachCte(\Var{id}, \Var{t}, \sigma, \theta, \Var{rt})}$ = \= 
    $\mathsf{\Var{id} \in \Var{t} \; \vee}$\\
    \> $\mathsf{\exists \; \Var{r} \in \Var{t}, \; 
      reachCte(\Var{id}, \sigma(\Var{r}), \sigma, \theta, \Var{rt})\; \vee}$\\
    \> $\mathsf{\exists}$ \=
    $\mathsf{\Var{tid} \in \Var{t},
      reachTable(\Var{id}, \Var{tid}, \sigma, \theta, \Var{rt})\; \vee}$\\
    
    \> $\mathsf{\exists \; \Var{cid} \in \Var{t},
      reachCte(\Var{id}, \theta(\Var{cid}), \sigma, \theta, \Var{rt})}$
\end{tabbing}
\vspace{-1em}  
  \caption{Reachability of a collectible element.}\label{fig:reachCte}
\end{figure}

Before presenting the predicate $\mathsf{reachCte}$, we introduce the predicate {\sf reachTable} (shown in \fref{reachCte}), 
which expands the reachability tree for a table \Var{tid}, when 
determining reachability for an identifier \Var{id} with respect
to a term \Var{rt}, and stores $\sigma$ and $\theta$.
%
%
We first check for reachability following references from a table
\Var{tid}, when it happens to be an ephemeron, as specified by the
predicate {\sf eph}.  This predicate says that $\Var{id}$ is reachable
from the value of a field (\Var{k}, $\Var{v}$) of an ephemeron table $\Var{tid}$
\Iff{} it is strongly reachable from $\Var{v}$, according to {\sf
  reachCte}, and the key \Var{k} cannot be GC, \ie is not a member of
\Var{cte} or it is reachable from the root set of references from
\Var{rt}: \ie the reachability of the value is affected by the reachability of
its key. In doing so, we must not take into account
$\mathsf{\Var{v}}$ to allow the collection of a field where the only
reference to the key comes from the value.  We use the notation
$\left.\theta\right|_{\left.\Var{tid}\right|_{\Var{k}}}$ to denote the
resulting store from removing the field with key \Var{k} from the
table \Var{tid}.

If the table is not an ephemeron, we just need to
consider each strong occurrence of a \NT{cte} present into the table, as defined
by \textsf{SO}.
Finally, for any table found during the expansion of the reachability tree, we
also need to look into its metatable, as it was the case when defining the 
predicate \textsf{reach}, in \sref{gc:simple}.






We now turn to the definition of {\sf reachCte} (also in \fref{reachCte}), which will have almost the same
signature as {\sf reach}, except for the addition of the term from which the 
root set of references is determined for the case of ephemeron tables.
As an aside, while it is possible to give a primitive or well-founded recursive definition, 
it would require cumbersome expressions for the recursive calls over stores of 
decreasing size. Instead, we followed~\cite{fsf} and defined the desired 
predicate as the least fixed point that satisfies the previous equation.

 
The predicate is defined assuming that the mere occurrence of a \Var{cte} into
\Var{t} implies that such value is strongly reachable. Recursive cases are 
defined such that they maintain this property of \Var{t}. The second disjunct 
dereferences references to values found into the term 
\Var{t}. Next we expand the reachability tree by following tables, as expressed
by {\sf reachTable}. The last disjunct checks into the environment of the closures
found during expansion, as in
Definition~\ref{Reachability Simple GC}.

\paragraph{GC cycle.}

Note that by enriching the notion of reachability with weak references, it could
be possible for the garbage collector to remove the binding of a table 
or closure identifier which is not strongly reachable but that is still 
present into a reachable weak table. This, of course, would generate dangling 
pointers if the program tries to dereference such identifiers through the weak
table.

If we forget about finalizers, we avoid such problems by simply replacing the predicate {\sf reach} in Definition~\ref{Simple GC cycle} of {\sf gc} with the new predicate {\sf reachCte}.
However, when considering finalizer, special care must be taken. We
therefore introduce a new predicate $\mathsf{gc_{fin\_weak}}$
(\fref{gcweak}), which is based on a modified $\sf{gc_{fin}}$ predicate. In
concrete, the new new predicate $\sf{gc'_{fin}}$ is a verbatim copy
of $\sf{gc_{fin}}$ but with the following changes:
\begin{enumerate}
\item We replace {\sf reach} with {\sf reachCte} in the {\sf fin}
  predicate.
\item We prevent for finalization to occur on a table that is also
  present as a value from a weak table by adding the following
  predicate:
\begin{align*}
\mathsf{not\_fin\_val(\Var{tid}, \theta)} & \doteq 
  \mathsf{\nexists \Var{tid'} \in \dom \theta, \Var{k} \in \NT{v} /}\\
  &\mathsf{(wk?(\Var{tid'}, \theta) \vee 
                           wv?(\Var{tid'}, \theta))} \wedge 
  \mathsf{(\Var{k}, \Var{tid}) \in \pi_1(\theta(\Var{tid'}))}  
\end{align*}
\item We remove the \emph{fin} portion of the predicate to let the new $\mathsf{gc_{fin\_weak}}$ predicate take care of marking the table with $\oslash$.
\end{enumerate}

\begin{figure}
\flushleft
$\mathsf{gc_{fin\_weak}(\Var{s},\sigma,\theta) = 
(\sigma_1',\theta_{1}''',\Var{s'})}$,
where 
$\mathsf{(\sigma_1', \theta_1', \Var{s'}) =  
gc'_{fin}(\Var{s}, \sigma, \theta)}$,
and:\\
$\mathit{wt}
 \left\{ 
   \begin{array}{lcl}
     \mathsf{\exists \theta_1'', dom(\theta_1'') = dom(\theta_1')}\\[.5em]
     
     \mathsf{\forall \Var{tid} \in dom(\theta_1''), 
                     \theta_1''(\Var{tid}) = \theta_1'(\Var{tid}) \; \vee}\\
     \mathsf{\;[\pi_1(\theta_1''(\Var{tid})) \subset 
               \pi_1(\theta_1'(\Var{tid}))\; 
     \wedge}\\
     \mathsf{\;\;\exists (k, v) \in \pi_1(\theta_1'(tid)),
                 (k, v) \notin \pi_1(\theta_1''(\Var{tid}))/}\\[.3em]
     \;\;\;\;\;\mathit{reach}
      \left\{ 
       \begin{array}{cc}
         \mathsf{wk?(\Var{tid}, \theta) \; \wedge \; \Var{k} \in \NT{cte} \; \wedge \; \neg reachCte(k, \Var{s}, \sigma, \theta, \Var{s})}\\
                 \vee\\ 
         \mathsf{wv?(\Var{tid}, \theta) \; \wedge \; \Var{v} \in \NT{cte} \; \wedge \; \neg reachCte(v, \Var{s}, \sigma, \theta, \Var{s})}\\
         \end{array}
      \right.\\[1.5em]
     
     \;\;\;\;\;\mathit{fin\;key}
      \left\{ 
       \begin{array}{lcl}
         \mathsf{wk?(\Var{tid}, \theta) \Rightarrow \neg marked(k, \theta)}\\
         \end{array}
          \right.\\[.5em]
     \;\;\;\;\;\mathit{rem}
      \left\{ 
       \begin{array}{lcl}
         \mathsf{\pi_2(\theta_1''(\Var{tid})) = \pi_2(\theta_1'(\Var{tid}))}\\ 
         \mathsf{\pi_3(\theta_1''(\Var{tid})) = \pi_3(\theta_1'(\Var{tid}))\quad]}
         \end{array}
     \right.
            \end{array}
          \right.$~\\~\\[1em]
$\mathit{fin}
\mathsf{
 \left\{ 
   \begin{array}{l}
     \mathsf{
     \theta'''_{1} = \theta''_1[\Var{tid} := 
     (\pi_1(\theta''_1(\Var{tid})), 
     \pi_2(\theta''_1(\Var{tid})),
     \oslash)]},
      \textit{if}\;\,
     \mathsf{\Var{s'} =  \Var{v}(\Var{tid}))}\\
     \qquad \vee \\
     \mathsf{\theta'''_{1} = \theta''_1,}
     \; \textit{if}\;
     \mathsf{\Var{s'} = \nil }
   \end{array}
 \right.}
 $
  \caption{GC cycle extended with weak tables.}
  \label{fig:gcweak}
\end{figure}

Essentially, after obtaining a new $\sf{\theta'_1}$ from ${\sf gc'_{fin}}$, the returned object store $\sf{\theta'''_1}$ might have a few discrepancies from that of $\sf{\theta'_1}$, since GC may remove fields of tables when
their keys or values are not strongly reachable.

More concretely, $\sf \theta'''_1$ is the store obtained from updating
$\sf \theta''_1$ after marking with $\oslash$ the table being
finalized, if applicable (\emph{fin}). And $\sf \theta''_1$ is
obtained from $\sf \theta'_1$ after noting that they have the same
domain (table ids), and for every table \Var{tid}, they either have
the same definition or the table in $\sf \theta'_1$ has a field
(\Var{k}, \Var{v}) that is not present in $\sf \theta'''_1$ and:
  \begin{itemize}[wide,labelindent=0pt,leftmargin=\parindent]
  \item[{\it reach}:] The field has a not strongly reachable key or
    value, depending on the table weakness. Note that we pass \Var{s}
    as the last argument of {\sf reachCte}, to preserve it as the root
    set of references from which any new expansion of the reachability
    tree must begin.
  \item[{\it fin\; key}:] In the case of weak keys susceptible for
    finalization, they are removed only after they are finalized. This
    restriction allows for a finalizer of a weak key to access any
    information associated with that key.
  \item[{\it rem}:] The remaining components of the internal
    representation of tables are not altered.
  \end{itemize}





Finally, there is no need for the redefinition of the GC step: the details of GC
of weak tables are all abstracted into the $\mathsf{gc_{fin\_weak}}$ metafunction, 
and its interference with the execution of the program does not differ from what 
regular GC does.



\section{Properties of GC}
\label{sec:properties_gc}
In this section we present the formal framework used to study
properties of our specification of GC. We conclude this section with
an important theorem about Correctness of GC, and in the way we
provide the necessary tools required to discuss about
non-deterministic computations; which form the foundation stone of
\gctool~(\sref{gctool}).




\subsection{Result of a program}

\begin{figure}
\begin{tabbing}
  $\mathsf{result}$\=$\mathsf{(\sigma : \theta : 
           \Var{E}\;[\![\KW{return} \; \Var{v}_1,...,\Var{v}_n]\!])}$ \=
  $\mathsf{= 
    \left.\sigma\right|_{S} :
    \left.\theta\right|_{T} : \KW{return} \; \Var{v}_1,...,\Var{v}_n}$\\[.2em]
  \quad $\mathit{where} \left\{ 
    \begin{array}{ll}
      \mbox{\Var{E} does not contain a return point}\\[.2em]
      \mathsf{S = \bigcup\limits_{\Var{r} \,\in\, \dom \sigma, R(\Var{r})} r}
      \\[.2em]
      \mathsf{T = \bigcup\limits_{\Var{id} \,\in\, \dom \theta, R(\Var{id})} \Var{id}}\\[.2em]
      \mathsf{R(i) = reach(i, \KW{return} \; \Var{v}_1,...,\Var{v}_n,
                                        \sigma, \theta)}
    \end{array}
  \right.$\\[.3em]

  $\mathsf{result(\sigma : \theta : \KW{\$err} \; \Var{v})}$ \>\>
  $\mathsf{= \left.\sigma\right|_{S} :
    \left.\theta\right|_{T} : \KW{\$err} \; \Var{v}}$\\
  \>$\mathit{where}~{\sf S}~\mathit{and}~{\sf T}~\mathit{are~defined~as~before,~but~with}$\\
  \> $\mathsf{R(i) = reach(i, \KW{\$err} \; \Var{v},
                                        \sigma, \theta)}$\\[.3em]

    $\mathsf{result(\sigma : \theta : \KW{;})}$ \>\>
    $\mathsf{= \emptyset : \emptyset : \KW{;}}$
  \end{tabbing}
\caption{Result of a program and associated functions.}
\flabel{result_function}
\end{figure}



We start by defining the notion of \emph{result} of a Lua
program. Essentially, it consists of the term from the last
configuration of its convergent computation, together with the
information from the stores needed to give meaning to the term's free
variables. That is, we strip off from the stores any information
irrelevant to the final computation of the program.

To capture the previous idea we use a function, {\sf result} 
(\fref{result_function}), that given a final configuration of a program it 
extracts the required information from the stores to explain the result 
represented by said configuration. In order to understand the different cases 
considered by the function, we must state a standard corollary of the progress 
property for our semantics, which explains the expected final configurations for 
$\luastep$ (that is, Lua without GC):
\begin{corollary}[Corollary of progress]
\label{progress_prop}
For every well formed configuration $\sigma : \theta : \Var{s}$,
just one of the following situations hold:
  \begin{itemize}
  \item The execution diverges, denoted $\sigma : \theta : s \; \Uparrow$.

  \item The execution ends with an error $\KW{error} \; \Var{v}$, and
    stores $\sigma'$ and $\theta'$, denoted 
    $\sigma : \theta : \Var{s} \; \Downarrow \sigma' : \theta' : 
     \KW{error} \; \Var{v}$.

  \item The execution ends normally, with stores $\sigma'$ and $\theta'$,
    and some values $\Var{v},...$ are returned:
    \begin{quote}
      $\sigma : \theta : s \; \Downarrow \sigma' : \theta' : 
      \Var{E} \; [\![\KW{return} \; \Var{v},...]\!]$, where \Var{E} does not
      contain the point to which the \KW{return} statement must jump.  
    \end{quote}

  \item The execution ends normally, with stores $\sigma'$ and $\theta'$, and 
    no value is returned:
    \begin{quote}
      $\sigma : \theta : s \; \Downarrow \sigma' : \theta' : \T{;}$
    \end{quote}
  \end{itemize}
  Where it corresponds, the resulting configuration is also well formed.
\end{corollary}

The condition expressed for the evaluation context \Var{E}, in the case of
a computation that ends in 
$\Var{E} \; [\![\KW{return} \; \Var{v},...]\!]$, implies that the \KW{return}
statements occurs outside of a function: it is the result returned by the 
program, which will be received, for example, in the host application where the
Lua program is embedded.

We omit the notion of well-formedness, as it is standard: it rules out not just 
ill-formed programs, but also ill-formed terms that represent intermediate 
computations. It express, mainly, restrictions that cannot be captured by our 
context-free grammar.

Coming back to the function \textsf{result}, it considers each possible final configuration, keeping only the 
bindings from the stores that are needed to completely describe the result. It uses the function {\sf reach} from \sref{gc:simple}. In the case of a \KW{return}, it strips out the context \Var{E}. Though
simple, in the context of syntactic GCs such notion of result 
is not be sensible to different syntactic GC strategies, or even to the 
complete absence of GC.



Computing the result of a program allows us to compare different runs from the  
same program. We 
%
assume that there exists an $\alpha$-conversion between locations from $\sigma$ 
and $\theta$, even when real programming languages often provide several library
services that may break $\alpha$-conversion. For example, in Lua it is possible 
to convert a table id to a string using the library service 
{\sf tostring}. Naturally, if we include this service, we would be able 
to write programs whose returned values will depend upon obscure details of 
memory management, and that will be beyond formal treatment for the purpose of
comparison of results. 
Thus, we assume that the semantics
of $\luastep$ is deterministic, which basically boils down to:
\begin{assumption}[Restrictions to the model]\label{ass:restrictions}~

\begin{enumerate}
\item The memory manager is deterministic, and new references are
  always created fresh, \ie there is no re-use of memory location.
\item There are no services that exposes external variables, like the time, the file system, a random number generator, \etc
\end{enumerate}
\end{assumption}
\noindent
The first assumption can be lifted off if services that
expose the details of memory management are prohibited (iterators for tables 
with non-numeric keys, the {\sf tostring} service, \etc). 

\subsection{Observations}


The standard sanity check of our specification of GC (without interfaces to the
garbage collector), consists in showing that the addition of a step of GC does
not change the semantics of the running program. In the context of our dynamic
semantics we capture this idea with a notion of \emph{observations} over 
programs.

We parameterize the definition over a relation $\rightarrow$ that
formalizes execution steps. For our studies, $\rightarrow$ will be
$\luastep$ (\ie our original model of Lua's operational semantics)
with or without GC steps.  We will reuse the notation introduced in
Corollary~\ref{progress_prop} to speak about the convergence of
computations, but now we will subscript with $\rightarrow$, to
indicate that we are computing using only the execution rules from
$\rightarrow$.  For brevity, we will use \Var{C} for a variable that
ranges over the set of configurations.
 




\begin{definition}[Observations]
For a given well-formed configuration \Var{C}, and execution rules 
$\rightarrow$:
\begin{quote}
$\mathsf{obs(\Var{C}, \rightarrow) = \{ \bot \; | \; \Var{C} \; 
                                             \Uparrow_{\rightarrow} \} \cup
                                     \{ result(\Var{C'}) \; | \;
                                        \Var{C} \; 
                                        \Downarrow_{\rightarrow}
                                        \Var{C'}\}}$
\end{quote}
\end{definition}

The previous definition hinges on the fact that a progress property holds for
$\rightarrow$: if {\sf result} is defined over the last configuration of a 
convergent computation, this configuration must be a valid final configuration. 
While this is true for $\rightarrow = \luastep$, we have not provided evidence 
that this is also the case after the addition of GC and its interfaces. Later,
in \sref{correctness_gc}, we will argue that by including $\gcstep$ to $\luastep$
we are not introducing stuck states.

Observations are useful to describe program equivalence:
\begin{definition}[Program equivalence]\label{program_equivalence}~\\[-1em]
\begin{quote}
$\mathsf{(\Var{C}, \rightarrow) \equiv 
 (\Var{C'}, \rightarrow')
 \Leftrightarrow
  obs(\Var{C}, \rightarrow) = 
  obs(\Var{C'}, \rightarrow')}$
\end{quote}
\end{definition}

\subsection{Garbage}
With the definitions developed so far we can now formalize a notion of garbage 
as a binding (a pair reference-value) that can be removed without changing the 
meaning of a program:

\begin{definition}[Garbage]
For a given well-formed configuration 
$\sigma \uplus \{ (\Var{r}, \Var{v}) \} : \theta : s$, operational semantics
$\rightarrow$, the binding $(\Var{r}, \Var{v})$ is \emph{garbage} with respect to 
$\rightarrow$, \Iff:
\begin{quote}
$(\sigma \uplus \{ (\Var{r}, \Var{v}) \} : \theta : s, \rightarrow) \equiv 
(\sigma : \theta : s, \rightarrow)$
\end{quote}
A binding from $\theta$ is defined as garbage in an analogous manner.
\end{definition}


The concepts introduced so far will allow us to define and study a notion of
correctness of GC in the absence of its interfaces, but also could be of 
use for future studies of properties and applications of the model for weak 
tables and finalization.
 

\subsection{Correctness of $\gcstep$.}
\label{sec:correctness_gc}

With the previously defined notions we can tackle the study of several desirable
properties of GC. For GC without interfaces we can perform its standard sanity
check, \ie to prove its soundness property: the addition of a GC step does not
change the semantics of a program. Informally, it consists in showing that by
adding $\gcstep$ the observations over a given program are not altered. The
desired statement is captured in the following statement (where
$\luagcstep = \luastep \cup \gcstep$):

\begin{theorem}[GC correctness]
For a given well-formed configuration 
$\sigma : \theta : \Var{s}$, 
\begin{center}
$(\sigma : \theta : \Var{s}, \luastep) \equiv
(\sigma : \theta : \Var{s}, \luagcstep)$
\end{center}
\end{theorem}
The proof is included in appendix \ref{app:properties_gc}. 

We obtain as corollary that $\luagcstep$ is deterministic:
\begin{corollary}[Determinism of GC]
\label{thm:singleton_under_luagcstep}
For a well-formed configuration 
$\mathsf{\sigma : \theta : \Var{s}},\;\mathsf{\vert obs(\sigma : \theta : \Var{s}, \luagcstep) \vert = 1}$
\end{corollary}

Naturally, after the introduction of weak tables or finalizers, programs may
no longer exhibit deterministic behavior, hence the requirement of a set of 
observations in order to be able to express the possible outcomes of a Lua 
program under $\mapsto$, the complete dynamic semantics: 

\begin{theorem}[Non-deterministic behavior]
Form some well-formed configuration $\sigma : \theta : \Var{s}$, 
\begin{center}
$(\sigma : \theta : \Var{s}, \luagcstep) \not \equiv
(\sigma : \theta : \Var{s}, \mapsto)$
\end{center}
\end{theorem}
\begin{proof}
Consider the program presented in 
\sref{introduction} (\fref{ex-non-det}).
\end{proof}

As a first attempt in recovering the deterministic behavior of programs that make
use of weak tables, in the next section we introduce \gctool.

\section{\gctool: ensuring GC-safeness}\label{sec:gctool}

As the code in \fref{ex-non-det} shows, a program using weak tables
could exhibit non-deterministic behavior. Nonetheless, given the
usefulness of weak tables to easily implement several data-structures
(\eg caches, weak sets, property tables)~\cite{ecwt}, it is important
to understand their semantics, and to have tools to prevent common
pitfalls in their use. In the first part of the present paper we
aimed at the former, and now we turn our attention into what
constitutes the first steps into the later.

More concretely, in this section we introduce \gctool, a prototype
static analyzer that aims at the detection of ill-uses of weak tables,
that could lead to non-deterministic behavior. We are mostly concerned
with access to fields of weak tables that are not strongly
reachable. While the general problem is known to be
undecidable~\citep{fswr}, we propose an approximation to the solution
by combining techniques from statics semantics (type inference, type
checking and data-flow analysis) together with weak tables'
semantics.

For a given Lua program \Var{p}, being $\mapsto$ the dynamic semantics
that includes $\luastep$ and GC with interfaces, if
$\mathsf{obs}(\emptyset : \emptyset : \Var{p}, \mapsto)$ is a
singleton we say that \Var{p} is \emph{gc-safe}, and denote with
$\safeset$ the set of gc-safe programs. In our approach we aim at
taking a user program and trying our best to guess if it belongs to
$\safeset$, without asking the user for modifications of the program
or to use weak tables according to some particular idioms, as proposed
in~\cite{fswr}.

As this is the first step taken in implementing \gctool, we assume
some restrictions---that we mention where relevant---on the Lua
programs under consideration. We expect in the future to increase the
analysis power of the tool.

\paragraph{Overview}
\begin{figure*}
\input{figures/tool.tex}
\caption{Design of \gctool.}
\label{fig:gctool_design}
\end{figure*}

\fref{gctool_design} shows the design of \gctool. As a first step, we take a user
program \Var{p}, an we infer the type of its local variables and function 
definitions. In particular, at this point we recognize if the 
evaluation of a given expression involves the access to a field of a weak table,
and to determine the kind of information such field contains. This is 
important to understand if the result of such evaluation could be unpredictable. The result of type inference is an annotated program $\typed{\Var{p}}$. 

The next step consists in the extraction of 
information useful to determine if the references in a given expression are reachable
from the root set of references. To compute 
the root set at some point of the program, we use a syntactic approximation 
consisting in the set of definitions of variables which are valid at that point. That is, we solve the problem of \emph{reachable definitions}~\citep{drag} for $\typed{\Var{p}}$ by constructing its 
\emph{control flow graph} (cfg), annotating each expression and statement of 
the program with the 
set of definitions that are valid at that particular point, obtaining $\rchdef{\Var{cfg}}$.

The last step consists in taking $\typed{\Var{p}}$ and $\rchdef{\Var{cfg}}$, and
performing type checking over $\typed{\Var{p}}$. In that way, we are able to 
reconstruct the type of complex expressions, and to recognize whether the
evaluation of a given expression involves the access to a field of a weak 
table. 
If it is the case, we will
query $\rchdef{\Var{cfg}}$ for the set of valid variable definitions at the
corresponding point of the program 
and
determine the reachability of the corresponding table field,
following the semantics of weak tables from \sref{weak_tables}.


In the reminder of this section we explain the different steps of
\gctool{} and present examples showing its potential.

\subsection{Type system}
\newcommand{\st}[1]{\ensuremath{\langle #1 \rangle}}
\newcommand{\vst}{\NT{v}\ensuremath{_{\sf st}}}

\begin{figure}
  \begin{tabbing}
  \textbf{Typed language}\\[.3em]
  \NT{s} ::=  \= \ldots{} \sep 
              \T{local} \NT{x} \T{:} \NT{t} \T{,} ... = \NT{e} \T{,} ... 
              \T{in} \NT{s} \T{end}\\
             
  \NT{e} ::= \ldots{} 
             \sep \NT{t} \T{function} ( \NT{x} \T{:} \NT{t} \T{,} \ldots{} ) 
                            \NT{s} 
                  \T{end}\\[.5em]

  \textbf{Types}\\[.3em]
  \NT{t} ::= \= \NT{prmt} \sep \NT{st} \sep \T{dyn} \sep \NT{t} $\rightarrow$ \NT{t} \sep
                \{ [ \NT{st} ] : \NT{t} ... \} \NT{wkness} \sep 
                $\mu$ \NT{y} . \NT{t} \sep \NT{t} $\times$ \NT{t} \sep \emptup\\
  \NT{prmt} ::= \T{nil} \sep \T{num} \sep \T{bool} \sep \T{str}\\
  \NT{st} ::= \st{ \vst{} : \NT{prmt} } \\
  \vst{} ::= \T{nil} \sep \NT{string} \sep \NT{boolean} \sep \NT{number}\\
  \NT{wkness} ::= \T{strong} \sep \T{wk} \sep \T{wv} \sep \T{wkv}
  \end{tabbing}
  \vspace{-1em}
  \caption{Grammar for typed terms.}
  \flabel{type_terms}
\end{figure}

Common to all the steps of \gctool{} lies the type system. 
\fref{type_terms} shows the language extended with type
annotations for local variables and function definitions. As for the types, we have primitive types, \NT{prmt}, where we include the
\T{nil} type, \T{num}bers, \T{bool}eans and \T{str}ings. Then, we have
\emph{singleton types} \NT{st}, which \emph{lift} to the level of
types a literal value \vst{} (\T{nil}, a number, a string, or a
boolean). Singleton types serve two purposes in our work: to allow us
to statically know which field of a table is being indexed, and
therefore to know if the access is valid or not; and to track the
changes in the weakness of tables at each call to
\setmetatable.


Besides primitive and singleton types we have the \T{dyn} supertype
for variables whose type cannot be properly inferred statically; function types,
\mbox{\NT{t} $\rightarrow$ \NT{t}}; table types, \{ [ \NT{st} ] :
\NT{t} ... \} \NT{wkness}, which include a tag (\NT{wkness})
indicating the weakness of the table, and are restricted to be indexed
by singleton types; recursive types, $\mu$\NT{y}. \NT{t}, to better
support common programming idioms using tables; and product types,
\mbox{\NT{t} $\times$ \NT{t}} and \emptup{} (the empty tuple), which we use to 
express the domains of functions, though they have many other roles in typing Lua
programs (see, for example,~\cite{typedlua}).

Types are ordered by a typical subtyping relation~$<:$, except for minor 
simplifications: \T{dyn} is the supertype of every type; every primitive type 
\NT{p} is the supertype of any \st{\NT{v} : \NT{p}}; subtyping for function and
recursive types will be reduced to reflexivity, for purposes of simplification of
type inference; table types are related by width, depth and permutation 
subtyping; and product types are covariant. As an important
remark, we do not take a tables' weakness into account for subtyping
in order to let the weakness of a table to change through a given
program.

\subsection{Type inference}









 
Our type inference algorithm is based mostly on ideas introduced in~\cite{ttij},
where it is presented a type inference algorithm for a language that includes 
some features of JavaScript. For reasons of brevity we will not cover its 
details. We refer the reader to the cited work and our mechanization with PLT 
Redex
.

Informally, the essence of the process consists in traversing the AST
of a given program, generating constraints over the type that we
should assign to each expression. These constraints are generated
observing the way in which expressions are used in the program. For
our purposes constraints relates types of terms, according to our
subtyping relation, and restricts the fields that a given table type
should have.

The solution proposed in~\cite{ttij}, which we follow, works in
steps. First, for each expression, a new type variable is constructed;
then, these variables are constrained.
Once a set of constraints \Var{Cs} is generated for a given program, the 
algorithm proceed by inferring new constraints from \Var{Cs}, for which it is 
guaranteed that, if a solution exist for \Var{Cs}, then the same solution solves
the newly inferred constraints. This step intends to make evident the existence 
of a solution or expose any inconsistency present among the constraints, showing 
the absence of a solution. The last step generates solutions for 
each constraint.

Our type inference algorithm follows the previously described process, with 
minor additions to tackle the problem of type inference given our subtyping 
relation, which is slightly more complex: 
we have the supertype \T{dyn}, tuple types and a slightly more complex subtyping 
relation for primitive types, since we also have singleton types. The main additions involve enriching the expressiveness
of the language to express constraints over types, and an added step that refines
the possible types that could be assigned to an expression, for the case of
primitive types.

\subsection{Computing the Control Flow Graph}

In order to compute the cfg for the program we follow traditional ideas from~\cite{drag} adapted to Lua code. The resulting
$\rchdef{\Var{cfg}}$ contains a family of sets of definitions of variables 
that are valid at every statement and expression of the program being typed. We 
identify each of such points with a context \Var{C}, that we need to update 
accordingly through the whole type checking process. Such contexts also serve to identify the exact point in the program
where the tool identified a potentially non-deterministic
behavior. $\rchdef{\Var{cfg}}$ is indexed by these contexts. For brevity we do not show its definition, but 
it can be seen in the mechanization accompanying this paper.

\subsection{Type checking}


\begin{figure}
\begin{mathpar}
\inference{
\Gamma_1, \rchdef{\Var{cfg}}, \plug{\Var{C}}{\; \hole \T{[} \Var{e}_2 \T{]} \;} 
\typee 
\Var{e}_1: \Gamma_2: \{ [\Var{st}_1]:\Var{t}_1\T{,} \ldots \}\;\T{strong}\\
\Gamma_2,\rchdef{\Var{cfg}}, \plug{\Var{C}}{\; \Var{e}_1 \T{[} \hole \T{]} \;}
\typee 
\Var{e}_2: \Gamma_3: \Var{st}_2\\
\mtch{\{ [\Var{st}_1]:\Var{t}_1\T{,} \ldots \}}
     {\{ \ldots{} \T{,} [\Var{st}_2]:\Var{t}_2\T{,} \ldots \}}
}
{\Gamma_1, \rchdef{\Var{cfg}},\Var{C} 
\typee  
\Var{e}_1 \T{[} \Var{e}_2 \T{]} : \Gamma_3 : \Var{t}_2}

\inference{
\Gamma_1, \rchdef{\Var{cfg}},\plug{\Var{C}}{\; \hole \T{[} \Var{e}_2 \T{]} \;} 
\typee 
\Var{e}_1: \Gamma_2: \{ [\Var{st}_1]:\Var{t}\T{,} \ldots \}\;\T{wv}\\
\Gamma_2,\rchdef{\Var{cfg}},\plug{\Var{C}}{\; \Var{e}_1 \; \T{[} \hole \T{]} \;} 
\typee 
\Var{e}_1: \Gamma_3: \Var{st}_2\\
\mtch{\{ [\Var{st}_1]:\Var{t}\T{,} \ldots \}}
     {\{ \ldots{} \T{,} [\Var{st}_2]:\Var{cte} \T{,} \ldots \} } \\
 \mathsf{reachCte}(\rchdef{\Var{cfg}}[\Var{C}], \Var{e}_1 \T{[} \Var{e}_2 \T{]}, 
          \Gamma_3)}
{\Gamma_1, \rchdef{\Var{cfg}},\Var{C} 
\typee
\Var{e}_1 \T{[} \Var{e}_2 \T{]} : \Gamma_3 : \Var{cte}}

\inference{
\Gamma_1, \rchdef{\Var{cfg}},\plug{\Var{C}}{\; \hole \T{[} \Var{e}_2 \T{]} \;} 
\typee 
\Var{e}_1: \Gamma_2: \{ [\Var{st}_1]:\Var{t}_1\T{,} \ldots \}\;\T{wv}\\
\Gamma_2,\rchdef{\Var{cfg}},\plug{\Var{C}}{\; \Var{e}_1 \; \T{[} \hole \T{]} \;} 
\typee 
\Var{e}_2: \Gamma_3: \Var{st}_2\\
\mtch{\{ [\Var{st}_1]:\Var{t}_1\T{,} \ldots \}}
     {\{ \ldots{} \T{,} [\Var{st}_2]:\Var{t}_2 \T{,} \ldots \}} \\
 \;\; \Var{t}_2 \notin \NT{cte}}
{\Gamma_1, \rchdef{\Var{cfg}},\Var{C} 
\typee 
\Var{e}_1 \T{[} \Var{e}_2 \T{]} : \Gamma_3 : \Var{t}_2}
\end{mathpar}
\vspace{-1em}
\caption{Type checking for table indexing.}
\flabel{type_check_table_indx}
\end{figure}

\begin{figure*}
\begin{mathpar}
\inference{
\Gamma_1(\Var{x}) = \{ [\Var{st}]:\Var{t}\T{,} \ldots \}\;\Var{wkness}_1\\
\Gamma_1, \rchdef{\Var{cfg}}, 
\plug{\Var{C}}{\; \setmetatable (\Var{x}, \hole) \;} 
\typee 
\Var{e}: \Gamma_2: 
\{ \ldots \T{,} [\st{\mathsf{``\_\_mode"}:\T{str}}] : \st{ s : \T{str}} \T{,} \ldots \}\;
\Var{wkness}_2\\
\textsf{`v'} \in \Var{s}\\
\Gamma_3 = \Gamma_2[\Var{x} : \{ [\Var{st}]:\Var{t}\T{,} \ldots \}\;\T{wv}]
}
{\Gamma_1, \rchdef{\Var{cfg}},\Var{C} 
\types  
\setmetatable (\Var{x}, \Var{e}) : \Gamma_3}

\inference{
\Gamma_1(\Var{x}) = \{ [\Var{st}_1]:\Var{t}_1\T{,} \ldots \}\;\Var{wkness}_1\\
\Gamma_1, \rchdef{\Var{cfg}}, 
\plug{\Var{C}}{\; \setmetatable (\Var{x}, \hole) \;} 
\typee 
\Var{e}: \Gamma_2: \{ [\Var{st}_2]:\Var{t}_2\T{,} \ldots \}\;\Var{wkness}_2\\
\not \mtch{\{ [\Var{st}_2]:\Var{t}_2\T{,} \ldots \}}
     {\{ \ldots \T{,} [\st{\mathsf{``\_\_mode"} : \T{str}}] : \st{s : \T{str}}\T{,} \ldots \}} \vee
\textsf{`k'},\textsf{`v'} \notin \Var{s}\\
\Gamma_3 = \Gamma_2[\Var{x} : \{ [\Var{st}_1]:\Var{t}_1\T{,}\ldots\}\;\T{strong}]
}
{\Gamma_1, \rchdef{\Var{cfg}},\Var{C} 
\types  
\setmetatable (\Var{x}, \Var{e}) : \Gamma_3}
\end{mathpar}
\vspace{-1em}
\caption{Type checking: \setmetatable.}
\flabel{type_check_setmeta}
\end{figure*}

For brevity, we focus on the peculiarities of determining gc-safeness. Type checking is described 
by the typing relations 
$\typee \subseteq \Gamma \times \rchdef{\Var{cfg}} \times C \times \NT{e} 
                  \times \Gamma \times \NT{t}$ and
$\types \subseteq \Gamma \times \rchdef{\Var{cfg}} \times C \times \NT{s} 
                  \times \Gamma$, (partially) described in 
\fref{type_check_table_indx} and \fref{type_check_setmeta}, respectively. 

We denote with $\Gamma$ the environments mapping variable identifiers with their
types. Since we are typing a dynamic language, the statements and
expressions could change this mapping because of assignments of the
same variable to values having different type. 
Therefore, the typing relation includes a second environment to reflect the
changes. In the typing rules we ensure that the the values assigned to
a certain variable have types related by subtyping.


The first rule in \fref{type_check_table_indx} shows the typing of a (non-weak) 
table indexing, $\Var{e}_1 \T{[} \Var{e}_2 \T{]}$. As mentioned, in this 
prototype we simplify type checking by assuming that each key is a literal value.
Nonetheless, this is enough to type check common idioms involving tables, in Lua.
Assuming that $\Var{e}_2$ can be 
successfully typed as $\Var{st}_2$, we look into the type of $\Var{e}_1$ for a 
field with the same type,
$\mtch{\{ [\Var{st}_1] : \Var{t}_1\T{,} \ldots \}} {\{ \ldots{} \T{,} 
[\Var{st}_2] : \Var{t}_2\T{,} \ldots \}}$. In that case, we successfully type the 
whole indexing expression as $\Var{t}_2$, carrying in $\Gamma_3$ the (possible) 
modifications to the original environment. 

The second rule in \fref{type_check_table_indx} shows the typing of the indexing 
of a table that has weak values. If we can determine that the value being indexed
belongs to the set of values that can be garbage-collected (\NT{cte}), we 
need to check for the reachability of the value, to ensure a deterministic 
behavior of the indexing. We query $\rchdef{\Var{cfg}}$ for the set of definitions of 
variables that reaches the point \Var{C} of the program (\ie the table indexing),
and then traverse that set of definitions to check for the reachability of
the exact expression $\Var{e}_1 \T{[} \Var{e}_2 \T{]}$ as expressed by the predicate $\mathsf{reachCte}$ from \sref{weak_tables} (properly adapted to work with the information from our static analysis).


The last rule in \fref{type_check_table_indx} shows the case of indexing a table that has weak values, but when the value being accessed 
does not belong to the set \NT{cte}. In this case there is no risk of non-determinism.

For the present prototype the case of tables with weak keys (ephemerons) is 
trivially solved, since for any table, all of its keys will not belong to the set
\NT{cte} of values that could be garbage collected. If we allow
also \NT{cte}s as keys, checking for determinism would proceed analogous to the
case of weak values, though with the expected modifications dictated by the
semantics of ephemerons.

Another requirement for our typing relations is for them to recognize and keep
track of changes in the weakness of a given table, as a result of calls to the
service \setmetatable. \fref{type_check_setmeta} shows the typing rules for
calls to this service. 
In the first rule we show the case when the
given metatable contains the corresponding field to inform about a change in the
weakness of the table. We therefore require the metatable to have a field with  key of singleton type 
$\st{\mathsf{``\_\_mode"} : \T{str}}$ and value of singleton type \st{\Var{s} : \T{str}}, with \Var{s} containing 
the character `{\sf v}'. The environment $\Gamma_3$ will contain the updated weakness of table
\Var{x}. The last rule shows the case of a call to \setmetatable{} with a 
metatable which does not contain proper information about changes in weakness of
the table: it will result in the table's weakness being set to \T{strong}, 
regardless of the original weakness of the table.

\subsection{Examples}

\begin{figure}
\begin{lstlisting}
local cache1 = {[1] = function() return 1 end,
                [2] = function() return 2 end,
                [3] = function() return 3 end}
local obj = {method = cache1[1], attr = {}}
local cache2 = {[1] = cache1[2]}
setmetatable(cache1, { __mode = "v"})
setmetatable(cache2, { __mode = "v"})
cache1[1]()
cache1[2]()
cache1[3]()
\end{lstlisting}  
\vspace{-1.5em}
\caption{Example: Implementation of a simple cache.}
\label{fig:ex-cache}
\end{figure}

\begin{figure}
\begin{lstlisting}
local t1 = {}
t1["attr1"] = 1
t1["method"] = function(x) return x + t1["attr1"] end
t1["attr2"] = (t1["method"] (t1["attr1"]))
setmetatable(t1, {__mode = "v"})
t1["method"](t1["attr2"])
\end{lstlisting} 
\vspace{-1.5em}
  \caption{Example: Tracking the addition of table fields.}
  \label{fig:ex-reach-tree}
\end{figure}

In this section we show the capabilities of code analysis of the
present version of \gctool{} with examples that, though artificial in
concept, are meant to pinpoint the possibilities of the proposed
approach.  As mentioned in the introduction, the program from
\fref{ex-non-det} is correctly flagged as non-deterministic.

Additionally, \fref{ex-cache} shows the implementation of a cache-like 
structure, \lstinline{cache1} in Line~1, as a table with weak values. This cache stores several closures in 
fields indexed by different numbers. Beginning from Line 4, we create  
weak and strong references to the closures stored in \lstinline{cache1}. In Line 
4 we create an object-like table, \lstinline{obj}, where we store a reference 
to one of the closures from \lstinline{cache1}. In Line 5 we define another 
cache-like table, \lstinline{cache2}, and we add another reference to a closure 
stored in \lstinline{cache1}. In lines 6--7 we set \lstinline{cache1} and 
\lstinline{cache2} to have weak values. What follows are accesses to the closures in 
\lstinline{cache1}, through indexing. \gctool{} correctly recognizes that the indexing in
Line 8 is safe, since it involves the access of a \NT{cte} (a closure), stored in a table with
weak values, but for which there is a strong reference coming from the presence
of the closure as a method from \lstinline{obj}. The situation is different for the last two accesses (lines 8 and 9): it recognizes two different kinds of ill accesses: in Line 9 the indexing involves a \NT{cte} value from a weak table,
but for which every reachability path contains at least one weak reference
(besides \lstinline{cache1}, it is only referenced from a value of 
\lstinline{cache2}): \ie it is not strongly reachable. In Line 10, the value accessed has no other reference besides the one from \lstinline{cache1}. 

In \fref{ex-reach-tree} we illustrate the possibility of keeping track of the
addition of new fields to tables, by means of assignments. The example features
a table, \lstinline{t1}, which is defined field by field, with every new field
defined in terms of the previous. The tool recognizes that, in the function call 
in Line 6, there is a \NT{cte} being accessed which is not strongly reachable. 
Also, type inference and type checking correctly solve the type of the parameter 
being passed in the call, which is not a \NT{cte}, hence, there is no risk of
non-deterministic behavior. The example also serves to showcase some of the 
constructions of the language that \gctool{} handles, which includes every 
syntactic form except functions returning multiple values, assignment and 
definitions of multiple variables and tables with \NT{cte}s keys.


\section{Future and related work}
\label{sec:relatedwork}


 
Future work could include one of the several venues of improvement of \gctool: 
an enriched type system, proofs of soundness of type inference and checking, 
and the inclusion of language features that were left out of this first prototype. The main known
drawback of using PLT Redex for this investigations is the poor 
performance of the resulting programs. The implementation is useful mainly for the testing ideas about 
static analysis, rather than tackling the analysis of real-world Lua programs.
Future work could include the re-implementation of \gctool{} in a more efficient 
language.

Another a promising line of work for the future is to adapt the core concepts to the
new ECMAScript, which includes weak references and finalizers~\citep{weak-js}.

As for related work, we group them in three: formalizations of GC,
theoretic tools related to the inference of types, and tools for
static analysis of GC.

\paragraph*{Formalizations of GC:} 
Leal \etal{} present in~\cite{fsf} a formal
semantics for a $\lambda$-calculus extended with references (strong and weak), 
and finalizers. From the literature surveyed, this is the only work where both 
interfaces to the GC are considered. The 
semantics presented for finalization does not impose an order of execution among
finalizers, and resurrected objects' semantics does not differ from live 
ones. Also, there is no interaction between weak references and finalization.
As described in~\sref{finalizers}, Lua's implementation of finalization imposes
a chronological order of finalization, and resurrected objects' semantics differs
from live ones in certain conditions, even with regard to resurrected objects
present in weak tables. This adds a certain level of interaction between 
finalization and weak tables.

Morrisett \etal{} present in~\cite{ammm} a reduction semantics for GC (named
$\lambda_{GC}$), but 
without any interface with the garbage collector. The theory developed for 
proving correctness for GC served as a major source of inspiration for our own 
development. The given specification for a GC cycle does not consider 
reachability, but rather observes for the appearance of free variables when
removing a given binding from the heap. In~\cite{acjro} is shown that 
specifying GC in terms of reachability results in an increased expressiveness
of the resulting model, reflected in the possibility of emulating even more 
trace-based GC strategies. We followed that path.

Donnelly \etal{} extended $\lambda_{GC}$ including weak 
references~\citep{fswr}. They use their model (named $\lambda_{\it weak}$) to tackle the 
semantics of the key/values weak references present in the GHC implementation 
of Haskell (a concept similar to ephemerons, also present in Lua). Also, 
they present a type system for their model and show how to use it in the 
collection of reachable garbage (\ie \emph{semantic} garbage). Finally, they 
tackle the problem of the introduction of non-determinism into the evaluation of 
a program that makes use of weak reference. They provide a decidable syntactic 
criterion for recognizing programs well-behaved with regard to GC (\ie with a 
deterministic behavior, regardless of their use of weak reference), and 
characterize semantically a larger class of programs with the same deterministic 
behavior. Because $\lambda_{\it weak}$ is directly derived form $\lambda_{GC}$ it 
lacks the expressiveness of a model based on reachability. On the other hand,
the theory developed for their model is based on a set of observations  
over programs that considers the possibility of a non-deterministic behavior.
 Being non-determinism a phenomenon also present in our model, their theory 
served as a source of inspiration for the development of ours.


The concept of ephemerons and their implementation in Lua is described 
in~\cite{ecwt}. However, they are not studied into a formal setting.

\paragraph*{Type inference for Lua:}
Type inference for Lua has being already tackled by Mascarenhas
\etal{} in~\cite{fabiophd} to obtain an optimized compiler for Lua
5.1.  In the same vein, Maidl \etal{} present Typed
Lua~\citep{typedlua}, a type system for Lua 5.2 that tackles several
of the complexities of the language, with special care in the typing
of common idioms used by the community of Lua.

While not strictly related to Lua, Anderson \etal{} introduce
in~\cite{ttij} a type inference algorithm for a language similar to
JavaScript, together with the formulation of several properties that
characterize the soundness of the proposed approach.

\paragraph*{Static analysis for GC:}
In~\cite{aspects-trace} the authors consider a form of local
static analysis to detect the type of reference (\emph{collectable},
\emph{weak}, and \emph{strong}) that occurs in a trace of
execution. They use this information to avoid
memory-leaks. In~\cite{region-based-mm} the static analysis performed
is used to determine the correct scope of weak and strong
references. Donnelly \etal{} propose in~\cite{fswr} the recognition of 
gc-safeness, first, by providing a restricted set of programs, characterized 
syntactically, for which it can be asserted their deterministic behavior, and 
later, by a semantic definition of a wider class of programs, though not 
recognizable through syntactic analysis. In contrast, our approach to gc-safeness
recognition aims at receiving the user program as it is, and doing a best-effort 
attempt in reasoning about the program's behavior.






\bibliographystyle{abbrv}
\bibliography{references}

\begin{appendices}
\section{Properties of GC}
\label{app:properties_gc}

To reach to a proof of the correctness of $\gcstep$ we will require, first, to 
check for several lemmas about simple properties that hold for both, $\luastep$
and $\gcstep$. 

\paragraph*{Properties preserved by $\luastep$.}
The first lemma states that once a binding becomes amenable 
for collection, it will remain in that state after any computation step from 
$\luastep$.\footnote{Note that such simple property does not hold anymore if
we introduce weak tables or finalization.} For its proof we will assume that the
reader is familiar with the model presented in \cite{dls}. A complete proof would
require case analysis on every computation step from said model. For reasons of 
brevity, we will consider just a few cases.

\begin{lemma}
\label{preserve_non_reach}
For configurations 
$(\sigma_1 : \theta_1 : \Var{s}_1)$, $(\sigma_2 : \theta_2 : \Var{s}_2)$, if 
$(\sigma_1 : \theta_1 : \Var{s}_1) \luastep (\sigma_2 : \theta_2 : \Var{s}_2)$, 
for $(\sigma_1 : \theta_1 : \Var{s}_1)$ well-formed, then
$\forall \Var{l} \in dom(\sigma_1) \cup dom(\theta_1)$,
$\neg \mathsf{reach}(\Var{l}, \Var{s}_1, \sigma_1,\theta_1)
\Rightarrow 
\neg \mathsf{reach}(\Var{l}, \Var{s}_2, \sigma_2 , \theta_2) 
$.
\end{lemma}
\begin{proof}
We will follow the modular structure of $\luastep$ to reason over 
the step that transforms $(\sigma_1 : \theta_1 : \Var{s}_1)$ into 
$(\sigma_2 : \theta_2 : \Var{s}_2)$. We have the following cases for the step
taken from $\luastep$:
\begin{itemize}
  \item[-] {\it The computation does not depend on the content of the stores
(\ie it does not change bindings from a store or dereferences locations):} then,
it can be seen, by case analysis on each computation rule, that such computation 
step does not introduce any reference into the instruction term. What could 
happen is that the root set is reduced, by deleting references present into 
$\Var{s}_1$. In any case, for a given 
$\Var{l} \in dom(\sigma_1) \cup dom(\theta_1)$, if 
$\neg \mathsf{reach}(\Var{l}, \Var{s}_1, \sigma_1,\theta_1)$ it must be the 
case that also $\neg \mathsf{reach}(\Var{l}, \Var{s}_2, \sigma_2,\theta_2)$. 
   \item[-] {\it The computation changes or dereferences locations from 
$\sigma_1$} :
for a given $\Var{l} \in dom(\sigma_1) \cup dom(\theta_1)$, such that
$\neg \mathsf{reach}(\Var{l}, \Var{s}_1, \sigma_1,\theta_1)$ let us assume that  
$\mathsf{reach}(\Var{l}, \Var{s}_2, \sigma_2 , \theta_2)$. To reason about the
statement, we would need to do case analysis on every possible computation step
that interacts with the values store. As an example, let us consider the implicit
dereferencing of references to $\sigma_1$.\footnote{In \cite{dls}, for purposes of
simplification of the desugared Lua code from test suites, we included implicit 
dereferencing of references to values, as done in \cite{get}.} Then it must be 
the case that $\Var{s}_1$ matches against the pattern 
$\NT{E} [\![\; \Var{r} \;]\!]$, for an evaluation context \NT{E} and a reference
\Var{r}, and the computation is:
\begin{center}
$\mathsf{\sigma_1\; \textbf{:}\; \theta_1\; \textbf{:}\; 
        \overset{\overset{\Var{s}_1}{=}}{\NT{E} [\![\; \Var{r} \;]\!]}
           \luastep\;
        \sigma_1\; \textbf{:}\; \theta_1\; \textbf{:}\; 
        \overset{\overset{\Var{s}_2}{=}}{\Var{E} [\![ \sigma_1(\Var{r})]\!]}
        }$
\end{center}
where both stores remain unmodified after the computation.
Then, the root set just changed by replacing \Var{r} by the references in
$\sigma_1(\Var{r})$. If $\neg 
\mathsf{reach}(\Var{l}, \Var{s}_1, \sigma_1,\theta_1)$ but\\ 
$\mathsf{reach}(\Var{l}, \Var{s}_2, \sigma_2 , \theta_2)$, this would mean that
\Var{l} is reachable from the references in $\sigma_1(\Var{r})$. But in 
$\Var{s}_1$, the references from $\sigma_1(\Var{r})$ were also reachable, making
\Var{l} reachable in $\Var{s}_1$, contradicting our hypothesis. Then, it must be 
the case that if\\ $\neg \mathsf{reach}(\Var{l}, \Var{s}_1, \sigma_1,\theta_1)$, 
\Var{l} remains unreachable in $(\sigma_2 : \theta_2 : \Var{s}_2)$.
    \item[-] {\it The computation changes or dereferences locations from 
$\theta_1$}:
let us assume that for a given $\Var{l} \in dom(\sigma_1) \cup dom(\theta_1)$,\\
$\neg \mathsf{reach}(\Var{l}, \Var{s}_1, \sigma_1,\theta_1) \; \wedge \; 
\mathsf{reach}(\Var{l}, \Var{s}_2, \sigma_2 , \theta_2)$. Again we will just 
analyze one case, among every computation that interacts with the store 
$\theta_1$. We will consider the rule that describes how tables are allocated in
$\theta_1$. Then, it must be the case that $\Var{s}_1$ matches against the pattern
$\NT{E} [\![\; \Var{t} \;]\!]$, for an evaluation context \NT{E} and a table
constructor \Var{t}, where every field haven been evaluated, making the table
ready for allocation. Then, the (simplified) computation is:
\begin{mathpar}	
  \inference{\mathsf{\Var{tid} \notin dom(\theta_1)}\\
    \mathsf{\theta_2 = (\Var{tid}, \;
                       (\Var{t}, 
                       \;\nil,
                       \;\bot)),\;
                       \theta_1}}
  {\mathsf{(\sigma_1\; : \; \theta_1\;:\;\NT{E} [\![\; \Var{t} \;]\!])
      \;\luastep\;
      (\sigma_1\; : \; \theta_2\;:\;\NT{E} [\![\; \Var{tid} \;]\!])}}
\end{mathpar}
where the values store remains unchanged, \ie $\sigma_2 = \sigma_1$. 
Then, the root set just changed by replacing the references in \Var{t} by the 
fresh table identifier \Var{tid}. If $\neg 
\mathsf{reach}(\Var{l}, \Var{s}_1, \sigma_1,\theta_1)$ but 
$\mathsf{reach}(\Var{l}, \Var{s}_2, \sigma_2 , \theta_2)$, this would mean 
that $\Var{l} = \Var{tid}$, which cannot be the case as
$\Var{l} \in dom(\sigma_1) \cup dom(\theta_1)$ and\\ 
$\Var{tid} \notin dom(\sigma_1) \cup dom(\theta_1)$. Then it must
be the case that if\\
$\neg \mathsf{reach}(\Var{l}, \Var{s}_1, \sigma_1,\theta_1)$, \Var{l} remains 
unreachable in $(\sigma_2 : \theta_2 : \Var{s}_2)$. 
\end{itemize}
\end{proof}

The following definition and lemma capture a standard concept in operational 
semantics for imperative languages: for a given instruction term, the outcome 
of its execution under given stores will depend on the content of the reachable 
portion of said stores.

\begin{definition}[]
For well-formed configurations 
$(\sigma_1 : \theta_1 : \Var{s})$ and $(\sigma_2 : \theta_2 : \Var{s})$,
we will say that both configurations \emph{coincide in the reachable portion
of their stores}, denoted
\begin{center}
$(\sigma_1 : \theta_1 : \Var{s}) \rcheq 
(\sigma_2 : \theta_2 : \Var{s})$
\end{center}

if and only if $\forall \Var{l} \in dom(\sigma_1) \cup dom(\theta_1) /
\mathsf{reach}(\Var{l}, \Var{s}, \sigma_1,\theta_1)$, then:
\begin{itemize}
  \item[-] $\mathsf{reach}(\Var{l}, \Var{s}, \sigma_2,\theta_2)$

  \item[-] $\Var{l} \in dom(\sigma_1) \Rightarrow 
         \sigma_1(\Var{l}) = \sigma_2(\Var{l})$

  \item[-] $\Var{l} \in dom(\theta_1) \Rightarrow 
        \theta_1(\Var{l}) = \theta_2(\Var{l})$
\end{itemize}
and the same holds $\forall \Var{l} \in dom(\sigma_2) \cup dom(\theta_2)$.
\end{definition}

In the previous definition, we are assuming that, if needed, it is always 
possible to provide a renaming of locations from both configurations to make them
equivalent in the sense expressed by $\rcheq$. Finally, it is easy to show that 
$\rcheq$ is an equivalence relation.

The important property, satisfied by configurations that coincide in the 
reachable portion of their stores, is stated in the following lemmas:

\begin{lemma}
\label{l_preserve_reach_eq}
For well-formed configurations
$(\sigma_1 : \theta_1 : \Var{s}_1)$ and $(\sigma_2 : \theta_2 : \Var{s}_1)$,
such that:
\begin{center}
{\small $(\sigma_1 : \theta_1 : \Var{s}_1) \rcheq 
(\sigma_2 : \theta_2 : \Var{s}_1)$}
\end{center} 
if $\; \exists (\sigma_3 : \theta_3 : \Var{s}_2)/ 
(\sigma_1 : \theta_1 : \Var{s}_1) \luastep 
(\sigma_3 : \theta_3 : \Var{s}_2)$, then\\
$\exists (\sigma_4 : \theta_4 : \Var{s}_2)/ 
(\sigma_2 : \theta_2 : \Var{s}_1) \luastep
(\sigma_4 : \theta_4 : \Var{s}_2)$ and:
\begin{center}
{\small
$(\sigma_3 : \theta_3 : \Var{s}_2) \rcheq 
(\sigma_4 : \theta_4 : \Var{s}_2)$}
\end{center}
\end{lemma}
\begin{proof}
We will follow the modular structure of $\luastep$ to reason over 
the step that 
transforms $(\sigma_2 : \theta_2 : s_1)$ into $(\sigma_4 : \theta_4 : s_2)$:
\begin{itemize}
  \item[-] {\it The computation does not change bindings from a store or 
dereferences locations} : then it must be the case that every information from 
the stores is already put into the instruction term $s_1$ so as to make the 
computation from $\luastep$ viable, without regard to the content of
the stores. Also, after the computation the stores are not modified. It implies 
that:
\begin{center}
{\small
$(\sigma_1 : \theta_1 : \Var{s}_1) \luastep
(\overset{\overset{\sigma_1}{=}}{\sigma_3} : 
 \overset{\overset{\theta_1}{=}}{\theta_3} : \Var{s}_2) \;
\wedge
\;
(\sigma_2 : \theta_2 : \Var{s}_1) \luastep
(\overset{\overset{\sigma_2}{=}}{\sigma_4} : 
 \overset{\overset{\theta_2}{=}}{\theta_4} : \Var{s}_2)$}
\end{center}
The root set of references in both configurations,\\
$(\sigma_3 : \theta_3 : \Var{s}_2)$ and $(\sigma_4 : \theta_4 : \Var{s}_2)$, is
the same. And, since
$(\sigma_1 : \theta_1 : \Var{s}_1) \rcheq (\sigma_2 : \theta_2 : \Var{s}_1)$
and the stores are not modified after the step from $\luastep$, it follows that
the reachable portion of the stores, from the root set defined by $\Var{s}_2$,
must coincide, according to $\rcheq$, in the configurations obtained after  
$\luastep$. Hence:
\begin{center}
{\small
$(\sigma_3 : \theta_3 : \Var{s}_2) \rcheq (\sigma_4 : \theta_4 : \Var{s}_2)$}
\end{center}
\item[-] {\it The computation changes or dereferences locations from $\sigma_1$}:
we would need to do case analysis on each computation that interacts with the 
value store. As an example, let us consider the implicit dereferencing of 
references to the values store. The hypothesis can be rewritten as:
\begin{center}
{\small
$(\sigma_1 : \theta_1 : \Var{s}_1) \luastep 
(\overset{\overset{\sigma_1}{=}}{\sigma_3} : 
 \overset{\overset{\theta_1}{=}}{\theta_3} : \Var{s}_2)$}
\end{center} 
where $\Var{s}_2$ contains the value associated with the reference dereferenced
by $\luastep$. Because:
\begin{center}
{\small
$(\sigma_1 : \theta_1 : \Var{s}_1) \rcheq (\sigma_2 : \theta_2 : \Var{s}_1)$}
\end{center}

the dereferencing operation will return the same result, if executed over 
$\sigma_2$. Then:
\begin{center}
$(\sigma_2 : \theta_2 : \Var{s}_1) \luastep 
 (\overset{\overset{\sigma_2}{=}}{\sigma_4} : 
 \overset{\overset{\theta_2}{=}}{\theta_4} : \Var{s}_2)$ 
\end{center}

Finally, because the stores are unmodified, after the step from $\luastep$, and
since the reachable portions of the stores in the original configurations 
coincide, according to $\rcheq$, then, it must be the case that the reachable 
portions of the stores obtained after $\luastep$ must also coincide, if we 
consider the same root of references. Hence:
\begin{center}
$(\sigma_3 : \theta_3 : \Var{s}_2) \rcheq (\sigma_4 : \theta_4 : \Var{s}_2)$
\end{center}
\item[-] {\it The computation changes or dereferences locations from $\theta_1$}:
we would need to do case analysis on each computation that interacts with 
$\theta_1$. As an example, let us consider table allocation. 
The hypothesis can be rewritten as:
\begin{center}
{\small
$\mathsf{(\sigma_1 : \theta_1 : \NT{E} [\![ \Var{t} ]\!])
      \;\luastep\;
      (\overset{\overset{\sigma_3}{=}}{\sigma_1} : 
      \overset{\overset{\theta_3}{=}}
      {\theta_1 \uplus \{ (\Var{tid}, \Var{t}, \nil, \bot) \}} :
      \NT{E} [\![ \Var{tid} ]\!])}$
}
\end{center}
Then we can assume that:
\begin{center}
{\small
$\mathsf{(\sigma_2: \theta_2:\NT{E} [\![ \Var{t} ]\!])
      \;\luastep\;
      (\overset{\overset{\sigma_4}{=}}{\sigma_2} : 
      \overset{\overset{\theta_4}{=}}
      {\theta_2 \uplus \{ (\Var{tid}, \Var{t}, \nil, \bot) \}}: 
      \NT{E} [\![ \Var{tid} ]\!])}$
}
\end{center}

where, if needed, we could apply a consistent renaming of tables' id
in $(\sigma_2 : \theta_2 : \NT{E} [\![ \Var{t} ]\!])$, such that 
it preserves its equivalence with 
$(\sigma_1 : \theta_1 : \NT{E} [\![ \Var{t} ]\!])$ and \Var{tid} is 
available as a fresh table identifier. Then, it follows immediately that:
\begin{center}
{\small
$(\sigma_3 : \theta_3 : \NT{E} [\![ \Var{tid} ]\!]) \rcheq
(\sigma_4 : \theta_4 : \NT{E} [\![ \Var{tid} ]\!])$}
\end{center}
\end{itemize}
\end{proof} 

Finally, the following lemma express an intuitive property that holds among 
final configurations that happen to be equivalent, according to $\rcheq$:

\begin{lemma}
\label{reach_eq_preserve_result}
For final configurations
$(\sigma_1 : \theta_1 : \Var{s})$ and $(\sigma_2 : \theta_2 : \Var{s})$,
such that
$(\sigma_1 : \theta_1 : \Var{s}) \rcheq (\sigma_2 : \theta_2 : \Var{s})$,
then:
\begin{center}
$\mathsf{result}(\sigma_1 : \theta_1 : \Var{s}) =
\mathsf{result}(\sigma_2 : \theta_2 : \Var{s})$
\end{center}
\end{lemma}
\begin{proof}
  The result will follow directly from the definition of {\sf result}, in 
  \fref{result_function}, and $\rcheq$. We will do a case analysis on the 
  structure of \Var{s}, for the configuration $(\sigma_1 : \theta_1 : \Var{s})$, 
  considering that it is the final state of a convergent computation:
  \begin{itemize}
    \item $\Var{s} = \KW{return} \Var{v}_1,...,\Var{v}_n$: for simplicity
     we consider the case $n = 1$, and we omit a possible context \Var{E} where
     the \KW{return} statement could occur, since it is not taken into account 
     by the notion of result of a program, as defined by $\sf result$. For larger 
     values of $n$ the reasoning remains the same:
     \begin{itemize}
       \item $\Var{v}_1 \in \NT{number} \cup \NT{string}$: then, neither 
         $\mathsf{result}(\sigma_1 : \theta_1 : \Var{s})$ nor 
         $\mathsf{result}(\sigma_2 : \theta_2 : \Var{s})$ depend on the content 
         of the stores. Hence,
           $\mathsf{result}(\sigma_1 : \theta_1 : \Var{s}) =
           \mathsf{result}(\sigma_2 : \theta_2 : \Var{s})$.
       
       \item $\Var{v}_1 \in \NT{tid} \cup \NT{cid}$: let us consider that 
         $\Var{v}_1 = \NT{tid}$ for some $\NT{tid} \in \dom {\theta_1}$ (the 
         reasoning for the case $\Var{v}_1 \in \NT{cid}$ is similar). Then, by 
         definition of {\sf result}:
         \begin{tabbing}
         $\mathsf{result(\sigma_1 : \theta_1 : \KW{return} \; \Var{tid})}$ \=
         $\mathsf{= 
           \left.\sigma_1\right|_{S} :
           \left.\theta_1\right|_{T} : \KW{return} \; \Var{tid}}$\\
         where 
         $\left\{ 
           \begin{array}{ll}
             \mathsf{S = \bigcup\limits_{\Var{r} \; \in \; \dom {\sigma_1},\;
             reach(\Var{r}, 
             \KW{return} \; \Var{tid},
             \sigma_1, \theta_1)}^{} r}
             \\
             \\
             \mathsf{T = \bigcup\limits_{\Var{id} \; \in \; \dom {\theta_1},\; 
             reach(\Var{id},
             \KW{return} \; \Var{tid},
             \sigma_1, \theta_1)}^{} \Var{id}}
           \end{array}
         \right.$
       \end{tabbing}

       Then, clearly 
       $\mathsf{result(\sigma_1 : \theta_1 : \KW{return} \; \Var{tid})}$ is
       depending on the reachable portions of $\sigma_1$ and $\theta_1$, 
       beginning with the root set defined by \Var{tid}. Because
       \begin{center}
         $(\sigma_1 : \theta_1 : \Var{tid}) \rcheq 
         (\sigma_2 : \theta_2 : \Var{tid})$
       \end{center} 
       the reachable portions of both configurations coincide. Hence,
       $\mathsf{result}(\sigma_1 : \theta_1 : \Var{s}) =
       \mathsf{result}(\sigma_2 : \theta_2 : \Var{s})$.
       
     \item $\Var{s} = \KW{error}~\Var{v}$: this case is identical to the 
       previous one.
       
     \item $\Var{s} = \T{;}$: then, 
       $\mathsf{result}(\sigma_1 : \theta_1 : \Var{s})$ is
       not depending on the content of the stores, and so is the case for\\
       $\mathsf{result}(\sigma_2 : \theta_2 : \Var{s})$. Hence,
       \begin{center}
         $\mathsf{result}(\sigma_1 : \theta_1 : \Var{s}) =
         \mathsf{result}(\sigma_2 : \theta_2 : \Var{s})$
       \end{center}

     \end{itemize}
  \end{itemize}
\end{proof}

\paragraph*{Properties preserved by $\gcstep$.}
A simple property to ask for is that reachable bindings are preserved, in the 
sense that they are still reachable and the value to which a given location is 
mapped is not changed after $\gcstep$. We express this property with the 
following two lemmas:

\begin{lemma}
\label{reach_bindings_remain}
For a well-formed configuration \mbox{$(\sigma_1 : \theta_1 : s)$}, if\\
$(\sigma_1 : \theta_1 : s) \gcstep (\sigma_2 : \theta_2 : s)$, for some
configuration \mbox{$(\sigma_2 : \theta_2 : s)$}, then
$\forall \Var{r} \in dom(\sigma_1), 
\mathsf{reach}(\Var{r}, \Var{s}, \sigma_1,\theta_1)
\Rightarrow \sigma_1(\Var{r}) = \sigma_2(\Var{r})$. The analogous holds for any
$\Var{id} \in dom(\theta_1)$.
\end{lemma}
\begin{proof}
Let $\Var{r} \in dom(\sigma_1), 
\mathsf{reach}(\Var{r}, \Var{s}, \sigma_1,\theta_1)$. 
Then, by Definition~\ref{Simple GC cycle} and $\gcstep$, it must be the case
that $\mathsf{gc(s, \sigma_1, \theta_1) = (\sigma_2, \theta_2)}$ and 
$\sigma_1(r) = \sigma_2(r)$.

For elements from $\mathsf{dom(\theta_1)}$ the reasoning is analogous to the 
previous case.
\end{proof}

\begin{lemma}
\label{gc_preserves_reach}
For a well-formed configuration \mbox{$(\sigma_1 : \theta_1 : s)$}, if\\
$(\sigma_1 : \theta_1 : s) \gcstep (\sigma_2 : \theta_2 : s)$, for some
configuration \mbox{$(\sigma_2 : \theta_2 : s)$}, then
$\forall \Var{l} \in dom(\sigma_1) \cup dom(\theta_1),
\mathsf{reach}(\Var{l}, \Var{s}, \sigma_1,\theta_1) \Rightarrow 
\mathsf{reach}(\Var{l}, \Var{s}, \sigma_2 , \theta_2)$.
\end{lemma}
\begin{proof}
We will prove it by induction on the minimum number of dereferences of 
locations from $\sigma_1$ or $\theta_1$ that needs to be performed to reach to a given location \Var{l}, for which
$\mathsf{reach}(\Var{l}, \Var{s}, \sigma_1,\theta_1)$ 
holds. By looking at Definition~\ref{Reachability Simple GC}, one of the following cases 
should hold:
\begin{itemize}
  \item[-] $\mathsf{\Var{l} \in \Var{s}}$: then it follows directly that
    $\mathsf{reach(\Var{l}, \Var{s}, \sigma_2 , \theta_2)}$.

  \item[-] $\mathsf{\exists \Var{r} \in dom(\sigma_1), 
   \Var{l} \in \sigma_1(\Var{r})}$, which is in a reachability path of minimum
   distance, from the root set to \Var{r} : then 
   $\mathsf{reach}(\Var{r}, \Var{s}, \sigma_1,\theta_1)$, and by inductive 
hypothesis, 
   $\mathsf{reach}(\Var{r}, \Var{s}, \sigma_2 , \theta_2)$. Also, by lemma 
   \ref{reach_bindings_remain}, $\sigma_1(\Var{r}) = \sigma_2(\Var{r})$. Then
   $\Var{l} \in \sigma_2(\Var{r})$ and
   $\mathsf{reach}(\Var{l}, \Var{s}, \sigma_2 , \theta_2)$,
   by definition.
  
  \item[-] $\mathsf{\exists \Var{tid} \in dom(\theta_1),
                   \Var{l} \in \pi_1(\theta_1(\Var{tid}))}$, which is in a
           reachability path of minimum distance, from the root set to \Var{l}: 
           then\\ 
           $\mathsf{reach}(\Var{tid}, \Var{s}, \sigma_1,\theta_1)$, and by 
           inductive hypothesis,\\ 
           $\mathsf{reach}(\Var{tid}, \Var{s}, \sigma_2 , \theta_2)$. Also, by 
           Lemma~\ref{reach_bindings_remain},
           $\theta_1(\Var{tid}) = \theta_2(\Var{tid})$. Then
           $\Var{l} \in \pi_1(\theta_2(\Var{tid}))$ and
           $\mathsf{reach}(\Var{l}, \Var{s}, \sigma_2 , \theta_2)$ by 
           definition.
  \item[-] $\mathsf{\exists \Var{cid} \in dom(\theta_1),
                   \Var{l} \in \theta_1(\Var{cid})}$, which is in a
           reachability path of minimum distance, from the root set to \Var{l}: 
           the reasoning is analogous to the previous case. It follows directly 
           that $\mathsf{reach}(\Var{l}, \Var{s}, \sigma_2 , \theta_2)$.
  \item[-] $\mathsf{\exists \Var{tid} \in dom(\theta_1), 
                    \Var{l} \in \pi_2(\theta_2(\Var{tid}))}$, which is in a 
           reachability path of minimum distance, from the root set to \Var{l}: 
           the situation is analogous to the previous case. It follows directly 
           that $\mathsf{reach}(\Var{l}, \Var{s}, \sigma_2 , \theta_2)$.
  \end{itemize}
\end{proof}

\begin{corollary}
\label{gc_preserve_reach_eq}
For well-formed configurations\\ 
$(\sigma_1 : \theta_1 : \Var{s})$ and $(\sigma_2 : \theta_2 : \Var{s})$,
if $(\sigma_1 : \theta_1 : \Var{s}) \gcstep 
(\sigma_2 : \theta_2 : \Var{s})$, then 
$(\sigma_1 : \theta_1 : \Var{s}) \rcheq 
(\sigma_2 : \theta_2 : \Var{s})$.
\end{corollary}
\begin{proof}
It is a direct consequence of lemmas \ref{reach_bindings_remain},
\ref{gc_preserves_reach} and the definition of $\rcheq$.
\end{proof}

While the following lemma directly refers to the notion of well-formedness of
configurations, it is not required to describe in detail such notion in order to
gain confidence about the following statement and its proof, since they are 
intuitive enough (for details about well-formedness, we refer the reader 
to~\cite{dls}).  Also, the lemma will allow us to 
extend the mentioned progress property for $\luastep$
to the semantics obtained adding  $\gcstep$. In particular, it will guarantee 
that the introduced notion of observations over programs is well-defined also
for $\luastep \cup \gcstep$, allowing us to state the desired correctness for
$\gcstep$.

\begin{lemma}
\label{gc_preserves_well-formed}
For a well-formed configuration \mbox{$(\sigma_1 : \theta_1 : s)$}, if
$(\sigma_1 : \theta_1 : s) \gcstep (\sigma_2 : \theta_2 : s)$, for some
configuration \mbox{$(\sigma_2 : \theta_2 : s)$}, then 
$(\sigma_2 : \theta_2 : s)$ is well-formed.
\end{lemma}
\begin{proof}
From the definition of $\gcstep$, it follows that the step does not change the
instruction term. Also, by the previous lemmas, it follows that $\gcstep$
does not introduce dangling pointers. They also state that $\gcstep$ does not
modify the stores in any other way, besides removing garbage.
Then, it must be the case that also $(\sigma_2 : \theta_2 : s)$ is well-formed. 
\end{proof}

\begin{lemma}
\label{gc_finite_steps}
Over a well-formed configuration $(\sigma : \theta : s)$, only a finite number of
$\gcstep$ steps can be applied.
\end{lemma}
\begin{proof}
By Definition~\ref{Simple GC cycle} and $\gcstep$, if
\begin{center}
$(\sigma : \theta : s) \gcstep (\sigma' : \theta' : s)$
\end{center}
then it must be the case that either $\sigma'$ or $\theta'$ is a proper subset 
of $\sigma$ or $\theta$, respectively. Then, being the stores partial finite 
functions, it is clear that GC can be performed at most a finite number of steps.
\end{proof}

The following lemma is a useful tool taken from \cite{ammm}. It codifies a 
simple intuition of plain GC: it must be possible to postpone any GC step, 
without changing the observations of the program. In its statement we use the 
fact that $\gcstep$ does not change the instruction term.

\begin{lemma}[Postponement]
\label{postponement}
For a given well-formed configuration $(\sigma_1 : \theta_1 : s_1)$, if
\begin{center}
$(\sigma_1 : \theta_1 : s_1) \gcstep (\sigma_2 : \theta_2 : s_1) 
\luastep
 (\sigma_3 : \theta_3 : s_2)$.
\end{center}
then $\exists (\sigma_4 : \theta_4 : s_2)$ such that:
\begin{center}
$(\sigma_1 : \theta_1 : s_1) \luastep (\sigma_4 : \theta_4 : s_2) \gcstep
 (\sigma_3' : \theta_3' : s_2)$
\end{center}
where $(\sigma_3 : \theta_3 : s_2) \rcheq (\sigma_3' : \theta_3' : s_2)$.
\end{lemma}
\begin{proof}
We will follow the modular structure of $\luastep$ to reason over 
the step that 
transforms $(\sigma_2 : \theta_2 : s_1)$ into $(\sigma_3 : \theta_3 : s_2)$:
\begin{itemize}
  \item[-] {\it The computation does not change bindings from a store or 
dereferences locations}: then it must be the case that every information from 
the stores is already put into the instruction term $s_1$ so as to make the 
computation from $\luastep$ viable, without regard to the content of the stores. 
Then, the hypothesis can be rewritten as:
\begin{center}
$(\sigma_1 : \theta_1 : s_1) \gcstep 
 (\sigma_2 : \theta_2 : s_1) \luastep
 (\sigma_2 : \theta_2 : s_2)$
\end{center}
If we take $(\sigma_4 : \theta_4 : s_2) = (\sigma_1 : \theta_1 : s_2)$, then 
we can assert that:
\begin{center}
$(\sigma_1 : \theta_1 : s_1) \luastep (\sigma_1 : \theta_1 : s_2) = 
 (\sigma_4 : \theta_4 : s_2)$
\end{center}
where we exploited the fact that, for the previous $\luastep$ step to be 
performed, the actual content of the stores does not affect the applicability 
and the outcome of said computation. 
Then, by Lemma~\ref{preserve_non_reach}, if a binding was ready to be collected
in $(\sigma_1 : \theta_1 : s_1)$ it will remain in that state  
in $(\sigma_1 : \theta_1 : s_2)$. So, by the non-deterministic nature of 
$\gcstep$, we could ask for it to remove the same bindings
that changed the stores from $(\sigma_1 : \theta_1 : s_1)$ into the stores
from $(\sigma_2 : \theta_2 : s_1)$. 
Hence, it must be the case that 
$(\sigma_1 : \theta_1 : s_2) \gcstep (\sigma_2 : \theta_2 : s_2)$ 
holds. We obtained:
\begin{center}
$(\sigma_1 : \theta_1 : s_1) \luastep (\sigma_1 : \theta_1 : s_2) \gcstep
 (\sigma_2 : \theta_2 : s_2)$
\end{center}

Finally, 
$(\sigma_3 : \theta_3 : \Var{s}_2) \rcheq (\sigma_3' : \theta_3' : \Var{s}_2)$
because 
\begin{center}
$(\sigma_3 : \theta_3 : \Var{s}_2) = (\sigma_2 : \theta_2 : \Var{s}_2) =
 (\sigma_3' : \theta_3' : \Var{s}_2)$
\end{center}

\item[-] {\it The computation changes or dereferences locations from $\sigma_1$}:
we would need to do case analysis on each computation that interacts with the 
value store. As an example, let us consider the implicit dereferencing of
a reference to $\sigma_1$. That is, the $\luastep$ step 
should be:
\begin{center}
$(\sigma_2 : \theta_2 : 
\overset{\overset{\Var{s}_1}{=}}{\Var{E} [\![ \Var{r}\;]\!]})
\luastep (\sigma_2 : \theta_2 : 
\overset{\overset{\Var{s}_2}{=}}{\Var{E} [\![ \sigma_2(\Var{r}) \;]\!]})$
\end{center} 
Then, the hypothesis can be rewritten as:
\begin{center}
$(\sigma_1 : \theta_1 : s_1) \gcstep 
 (\sigma_2 : \theta_2 : s_1) \luastep
 (\sigma_2 : \theta_2 : s_2)$
\end{center}
If we take $(\sigma_4 : \theta_4 : s_3) = (\sigma_1 : \theta_1 : s_2)$, 
we can assert that: 
\begin{center}
$(\sigma_1 : \theta_1 : 
\overset{\overset{\Var{s}_1}{=}}{\Var{E} [\![ \Var{r}\;]\!]})
\luastep (\sigma_1 : \theta_1 : 
\overset{\overset{\Var{s}_2}{=}}{\Var{E} [\![ \sigma_2(\Var{r}) \;]\!]})$
\end{center} 
because \Var{r} is reachable in $(\sigma_1 : \theta_1 : s_1)$, and
the $\gcstep$ step from the hypothesis preserves its binding, in the sense 
expressed in Lemma~\ref{reach_bindings_remain}: hence, if it was possible to 
perform the dereferencing in $(\sigma_2 : \theta_2 : s_1)$ (by hypothesis),
it must be possible to perform it in $(\sigma_1 : \theta_1 : s_1)$, obtaining
the same result. Finally, by preservation of bindings ready for collection after
a $\luastep$ step, Lemma~\ref{preserve_non_reach}, and the non-deterministic 
behaviour of $\gcstep$, we could ask for the GC step to remove exactly the 
necessary bindings so that 
$(\sigma_1 : \theta_1 : \Var{s}_2) \gcstep (\sigma_2 : \theta_2 : \Var{s}_2)$ 
holds. We obtained:
\begin{center}
$(\sigma_1 : \theta_1 : \Var{s}_1)
\luastep (\sigma_1 : \theta_1 : \Var{s}_2)
 \gcstep (\sigma_2 : \theta_2 : \Var{s}_2)$
\end{center}

Finally, 
$(\sigma_3 : \theta_3 : \Var{s}_2) \rcheq (\sigma_3' : \theta_3' : \Var{s}_2)$
because 
\begin{center}
$(\sigma_3 : \theta_3 : \Var{s}_2) = (\sigma_2 : \theta_2 : \Var{s}_2) =
 (\sigma_3' : \theta_3' : \Var{s}_2)$
\end{center}

\item[-] {\it The computation changes or dereferences locations from $\theta_1$}:
we would need to do case analysis on each computation that interacts with 
$\theta_1$. As an example, let us consider table allocation. 
Then, the hypothesis can be rewritten as:
\begin{center}
{\small
$(\sigma_1 : \theta_1 : s_1) \gcstep 
 (\sigma_2 : \theta_2 : s_1) \luastep
 (\sigma_2 : \theta_2 \uplus \{ (\Var{tid}, \Var{t}) \} : 
 \overset{\overset{\Var{s}_2}{=}}{\Var{E} [\![ \Var{tid} \;]\!]})$}
\end{center}
for an adequate internal representation of a table, \Var{t}, and table 
identifier \Var{tid}, that, for our purposes, it will be useful if\\
$\Var{tid} \notin dom(\theta_1)$. If it is not the case, we can continue with our
reasoning over an appropriate $\alpha$-converted configuration, where the 
references in $(\sigma_1 : \theta_1 : s_1)$ are consistently changed so as to 
make $\Var{tid} \notin dom(\theta_1)$. It is because of cases like this one 
that we cannot assert a stronger postponement statement, as the one in 
\cite{ammm}: we are not talking about convergence towards a single configuration;
we need to think in terms of $\rcheq$-equivalent configurations.

If we take
\begin{center} 
$(\sigma_4 : \theta_4 : s_3) 
 = (\sigma_1 : \theta_1 \uplus \{ (\Var{tid}, \Var{t}) \} : s_2)$
\end{center}

we know that:
\begin{center} 
$(\sigma_1 : \theta_1 : s_1) \luastep 
(\sigma_1 : \theta_1 \uplus \{ (\Var{tid}, \Var{t}) \} : s_2)$
\end{center}
where we can ask for the instruction term to be exactly 
$\Var{s}_2 = \Var{E} [\![ \Var{tid} \;]\!]$. By Lemma~\ref{preserve_non_reach} we know that every binding 
which is ready for collection in $(\sigma_1 : \theta_1 : s_1)$ is in the same 
state in $(\sigma_1 : \theta_1 \uplus \{ (\Var{tid}, \Var{t}) \} : s_2)$. Even 
more, such bindings just belongs to $\sigma_1$ or $\theta_1$. Then, by the 
non-deterministic nature of $\gcstep$ we could ask for it to remove just the 
necessary bindings so as to make true
\begin{center} 
$(\sigma_1 : \theta_1 \uplus \{ (\Var{tid}, \Var{t}) \} : s_2) 
\gcstep
(\sigma_2 : \theta_2 \uplus \{ (\Var{tid}, \Var{t}) \} : s_2)$.
\end{center}
Then, the following holds:
\begin{center} 
{\small
$(\sigma_1 : \theta_1 : s_1) \luastep ...  \gcstep
(\sigma_2 : \theta_2 \uplus \{ (\Var{tid}, \Var{t}) \} : s_2)$}
\end{center}

Finally, 
$(\sigma_3 : \theta_3 : \Var{s}_2) \rcheq (\sigma_3' : \theta_3' : \Var{s}_2)$
because 
\begin{center}
$(\sigma_3 : \theta_3 : \Var{s}_2) = 
 (\sigma_2 : \theta_2 \uplus \{ (\Var{tid}, \Var{t}) \} : \Var{s}_2) =
 (\sigma_3' : \theta_3' : \Var{s}_2)$
\end{center}
\end{itemize}
\end{proof}

\paragraph*{Correctness of simple GC}

The expected statement of GC correctness should mention that, for a given 
configuration, the observations under $\luastep$ should be the same that those
under $\luagcstep$ (\ie $\luastep \cup \gcstep$). However, under $\luastep$ and
$\luagcstep$ we expect the observations to be just a singleton:
the programs either diverge or reach to a end, returning some results or an error
object. Giving this observation, we could change the statement of GC correctness
to reach to a property that can be proved with less effort: given a 
configuration, under $\luastep$ its execution reaches to a end, if and only if 
its execution reaches to an end under $\luagcstep$, and, in both cases, what is 
returned (either values or error objects) is the same.

The stated property will imply the preservation of observations, as defined in Definition~\ref{program_equivalence}, but it will allow us to focus just on convergent 
computations; preservation of divergent computations will be a consequence of
the double implication structure of the statement:

\begin{theorem}[GC correctness]
\label{gc_correctness}
For a given well-formed configuration 
$\sigma : \theta : \Var{s}$, 
\begin{center}
$(\sigma : \theta : \Var{s}) \Downarrow_{\luastep} (\sigma' : \theta' : \Var{s}')
\Leftrightarrow
(\sigma : \theta : \Var{s}) \Downarrow_{\luagcstep} 
(\sigma'' : \theta'' : \Var{s}'')$
\end{center}
and
$\mathsf{result}(\sigma' : \theta' : \Var{s}') = 
 \mathsf{result}(\sigma'' : \theta'' : \Var{s}'')$.
\end{theorem}
\begin{proof}
Let us assume that 
$(\sigma : \theta : \Var{s}) \Downarrow_{\luastep} 
(\sigma' : \theta' : \Var{s}')$. Then
$(\sigma' : \theta' : \Var{s}')$ is a final configuration where \textsf{result}
is defined. Because $\luastep \; \subseteq \; \luagcstep$, 
it is always possible to emulate the previous trace by not using $\gcstep$ steps.
Then, $(\sigma : \theta : \Var{s}) \Downarrow_{\luagcstep} 
       (\sigma' : \theta' : \Var{s}')$, where it follows that, in both cases,
the computations returns the same, under $\luagcstep$ and $\luastep$.

On the other hand, let us assume that
\begin{center} 
$(\sigma : \theta : \Var{s}) \Downarrow_{\luagcstep} 
(\sigma' : \theta' : \Var{s}')$
\end{center}
Then, it must be the case that there exist a finite trace of computation steps, 
as follows:
\begin{center}
{\small
$(\sigma : \theta : s) \luagcstep
(\sigma_1 : \theta_1 : s_1) 
\luagcstep ... \luagcstep 
(\sigma_n : \theta_n : s_n)$
}
\end{center}
where $(\sigma_n : \theta_n : s_n) = (\sigma' : \theta' : s')$ is a final 
configuration over which \textsf{result} is defined.

By applying inductive reasoning over the number of computation steps and the 
Postponement Lemma~\ref{postponement}, it can be shown that we can rewrite the previous trace as 
follows:

\begin{center}
{\small
$(\sigma : \theta : s) \luastep ...
\luastep (\sigma_{i'} : \theta_{i'} : s_{i'})
\gcstep ...
\gcstep (\sigma_{n'} : \theta_{n'} : s_{i'})
$
}
\end{center}

where every computation that does not involve GC is performed at the beginning.
We obtained a convergent trace consisting only in $\luastep$ steps. That is:
\begin{center} 
$(\sigma : \theta : \Var{s}) \Downarrow_{\luastep} 
(\sigma_{i'} : \theta_{i'} : s_{i'})$
\end{center}
What remains is to see if the result is also preserved. To that end, note that
the postponement lemma used also tells us that
\begin{center}
$(\sigma_{n'} : \theta_{n'} : s_{i'}) \rcheq
(\sigma_{n} : \theta_{n} : s_{n})$
\end{center}
 
Then, because final configurations which are $\rcheq$ represent the same
result, according to Lemma~\ref{reach_eq_preserve_result}, it follows that
\begin{center}
{
$\mathsf{result}(\sigma_{n'} : \theta_{n'} : s_{i'}) = 
\mathsf{result}(\sigma_{n} : \theta_{n} : s_{n})$
}
\end{center}

Finally, because $\rcheq$ is closed under $\gcstep$ steps, Lemma~\ref{gc_preserve_reach_eq}, it must be the case that:

\begin{center}
{\small
$(\sigma_{i'} : \theta_{i'} : s_{i'}) \rcheq (\sigma_{n'} : \theta_{n'} : s_{i'})$
}
\end{center}

Hence,
\begin{center}
{
$\mathsf{result}(\sigma_{i'} : \theta_{i'} : s_{i'}) = 
\mathsf{result}(\sigma_{n'} : \theta_{n'} : s_{i'}) = 
\mathsf{result}(\sigma_{n} : \theta_{n} : s_{n})$
}
\end{center}

\end{proof}

An immediate corollary of the previous theorem is that, under $\luagcstep$, 
the set of observations over programs is a singleton, even under the 
non-determinism nature of $\gcstep$:

\begin{corollary}
\label{singleton_under_luagcstep}
For a well-formed configuration 
$\mathsf{\sigma : \theta : \Var{s}}$, \\
$\mathsf{\vert obs(\sigma : \theta : \Var{s}, \luagcstep) \vert = 1}$
\end{corollary}
\begin{proof}
It follows immediately from the previous theorem and the determinism of programs
under $\luastep$.
\end{proof}

Now, based on the observations of the beginning of this section, we can state an 
equivalent version of correctness for simple GC, but in terms of the notion
of observations previously defined:
\begin{corollary}[GC correctness]
For a given well-formed configuration 
$\sigma : \theta : \Var{s}$, 
\begin{center}
$(\sigma : \theta : \Var{s}, \luastep) \equiv
(\sigma : \theta : \Var{s}, \luagcstep)$
\end{center}
\end{corollary}
\begin{proof}
It follows directly from the previous corollary, together
with Theorem~\ref{gc_correctness}. Then, if 
$\mathsf{result(\sigma', \theta', \Var{s}') \in
         obs(\sigma : \theta : \Var{s}, \luastep)}$, for
$(\sigma : \theta : \Var{s}) \Downarrow_{\luastep} 
 (\sigma' : \theta' : \Var{s}')$, by theorem \ref{gc_correctness}, the previous
occurs if and only if $(\sigma : \theta : \Var{s}) \Downarrow_{\luagcstep} 
           (\sigma'' : \theta'' : \Var{s}'')$, where
\begin{center}
$\mathsf{result(\sigma', \theta', \Var{s}') = 
         result(\sigma'', \theta'', \Var{s}'')}$ 
\end{center}
Hence 
$\mathsf{result(\sigma', \theta', \Var{s}') \in
         obs(\sigma : \theta : \Var{s}, \luagcstep)}$, and we can conclude
that $\mathsf{obs(\sigma : \theta : \Var{s}, \luastep) =
             obs(\sigma : \theta : \Var{s}, \luagcstep)}$.
The converse is analogous.

If $\mathsf{\bot \in obs(\sigma : \theta : \Var{s}, \luastep)}$, by correctness
of GC, it must happen if and only if 
$\mathsf{\bot \in obs(\sigma : \theta : \Var{s}, \luagcstep)}$, and because of
the determinism of both, $\luagcstep$ and $\luastep$, we can conclude that:
\begin{center}
$\mathsf{obs(\sigma : \theta : \Var{s}, \luastep) =
         obs(\sigma : \theta : \Var{s}, \luagcstep)}$
\end{center}
The converse is analogous.
\end{proof}

\end{appendices}

\end{document}